\documentclass{article}
\usepackage{amsthm,amsmath,amsfonts,amssymb,natbib}
\RequirePackage[colorlinks,citecolor=blue,urlcolor=blue]{hyperref}
\usepackage[utf8]{inputenc}
\usepackage{xcolor}
\usepackage{authblk}
\newtheorem{thm}{Theorem}[section]
\newtheorem{prop}[thm]{Proposition}
 \newtheorem{lem}[thm]{Lemma}
\newtheorem{cor}[thm]{Corollary}
\newtheorem{prp}[thm]{Proposition}

 \theoremstyle{definition}
\newtheorem{rem}[thm]{Remark}
\newtheorem{asm}[thm]{Assumption}
\newtheorem{exm}[thm]{Example}
\newtheorem{dfn}[thm]{Definition}

\numberwithin{equation}{section}
\def\reff#1{{\rm(\ref{#1})}}

\def\be{\begin{equation}}
\def\ee{\end{equation}}

\def\bea{\begin{eqnarray}}
\def\eea{\end{eqnarray}}

\def\bea*{\begin{eqnarray*}}
\def\eea*{\end{eqnarray*}}

\newcommand{\cA}{\mathcal{A}}
\newcommand{\cB}{\mathcal{B}}
\newcommand{\cC}{\mathcal{C}}
 \newcommand{\cD}{\mathcal{D}}
\newcommand{\cE}{\mathcal{E}}
\newcommand{\cF}{\mathcal{F}}

\newcommand{\cH}{\mathcal{H}}
\newcommand{\cI}{\mathcal{I}}

\newcommand{\cL}{\mathcal{L}}
\newcommand{\cM}{\mathcal{M}}

 \newcommand{\cP}{\mathcal{P}}
 \newcommand{\cQ}{\mathcal{Q}}
 \newcommand{\cR}{\mathcal{R}}
 \newcommand{\cS}{\mathcal{S}}

\newcommand{\cZ}{\mathcal{Z}}

\newcommand{\Om}{{\Omega}}
\newcommand{\om}{{\omega}}

\newcommand{\eps}{{\epsilon}}

\def \E{\mathbb{E}}
\def \F{\mathbb{F}}

\def \L{\mathbb{L}}
 
 \def \N{\mathbb{N}}
 \def\P{\mathbb{P}}
\def \Q{\mathbb{Q}}
\def \R{\mathbb{R}}

\def \pointleq{\leq_{\Omega}}	
\def \pointgeq{\geq_{\Omega}}	

\def \oleq{\leq} 				
\def \ogeq{\geq} 				

\def \ppeq{\preceq}				
\def \mkt{(\cH,\tau, \oleq,\cI,\cR)}		

\def\bdd{\cB_b}
\def\bdl{\cB_\ell}

\definecolor{franks}{rgb}{.1,0,.7}

\renewcommand{\cite}[1]{\citet{#1}}

\begin{document}

\title{Viability and Arbitrage under Knightian Uncertainty
\thanks{Comments by Rose--Anne Dana, Filipe Martins-da-Rocha, and seminar and workshop participants at Oxford University, the  Institut Henri Poincaré Paris, and Padova University,  are  gratefully acknowledged.}}

\author[1]{Matteo Burzoni\thanks{Matteo Burzoni acknowledges the support of the Hooke Research Fellowship from the University of Oxford.}}
\author[2]{Frank Riedel\thanks{Frank Riedel gratefully acknowledges the  financial 
support of the German Research Foundation (DFG) via  CRC 1283.
 }}
 \author[3]{Mete Soner\thanks{Matteo Burzoni and Mete Soner  acknowledge the  support
 		of  the ETH Foundation, the Swiss Finance Institute and the Swiss
 		National Foundation through 
 		SNF $200020_-172815$.}}
\affil[1]{\small Universit\`a degli studi di Milano, Italy}
\affil[2]{\small Bielefeld University, Germany and University of Johannesburg, South Africa}
\affil[3]{\small Department of Operations Research and Financial Engineering,
	Princeton University}

\maketitle
\begin{abstract}
We reconsider the microeconomic foundations of financial economics.  Motivated by   the importance of 
Knightian Uncertainty in markets, we present a model that  does not carry any probabilistic structure ex ante, yet is based on a common order.  We  derive the fundamental equivalence of economic viability of asset prices and absence of arbitrage. We also obtain a modified version of the Fundamental Theorem of Asset Pricing using the notion of  
\emph{sublinear} pricing measures. Different versions of the  Efficient Market Hypothesis are related to the assumptions one is willing to impose on the common order.     
\end{abstract}

\medskip
{\footnotesize{ \it Keywords:} Viability, Knightian Uncertainty, No Arbitrage, Robust Finance
 

}

\section{Introduction}
\label{s.introduction}
Asset pricing models  typically  take a basic set of securities as given and determine the range of  option prices  that is  consistent with the absence of arbitrage.  From an economic point of view, it is crucial to know if modeling security prices directly  is justified; an asset pricing model  is called viable if its   security prices can be thought of as (endogenous)  equilibrium outcomes of a competitive    economy. Traditional finance models rely on a probabilistic framework. 
The Capital Asset Pricing Model assumes that agents have mean-variance preferences and share the same view of mean and variance-covariance matrix of asset returns. The Consumption-Based Capital Asset Pricing Model derives asset returns from economic equilibrium and also assumes that agents share the same prior. 
\cite{HarrisonKreps79} realized that a    reference probability that  determines  the null sets, the topology, and the    order of the model is sufficient to prove the equivalence of economic viability and pricing by arbitrage.  The common prior or the weaker reference probability assumption is   made in most asset pricing models.

Recently, a large and increasing body of literature has focused on  decisions, markets, and  economic interactions under uncertainty. Frank Knight's pioneering work  (\cite{Knight}) distinguishes  \emph{risk} --  a situation that allows for an objective probabilistic description --   from   \emph{uncertainty} --   a situation that cannot be modeled by a
single probability distribution.
By now,  it is widely acknowledged that  drift and  volatility of asset prices, the term structure of interest rates, and credit risk are important instances in which the probability distribution of the relevant parameters  is imprecisely known, if  not completely unknown.  \cite{EpsteinJi13}   emphasize  the relevance of non-probabilistic uncertainty in financial modeling when parameters vary too frequently to be estimated accurately, or when nonlinearities arise that are too complex to be captured by existing models, or when non-stationarities prevent  the use of the law of large numbers or the central limit theorem.They show that a probability space framework is not able to model ambiguity about volatility\footnote{ Compare also the general conceptual discussion of ambiguity and limitations of probabilistic modeling  in \cite{LoMueller2010}. }.   

We take these insights as a motivation to reconsider the foundations of arbitrage pricing  and its relation to economic equilibrium  without imposing  a priori a probability space framework. We show that the basic relations between economic equilibrium (viability), absence of arbitrage, and suitable pricing functionals can be proved with ease by merely assuming a common order with respect to which   preferences are monotone\footnote{An obvious and intuitive example of an order that we have in mind here is the pointwise order, i.e.  when agents will prefer a contingent consumption plan over their endowment if the new plan pays off more in every state of the world. Preferences over monetary outcomes are naturally assumed to be monotone with respect to this basic order.}. Using this approach we achieve a unifying theory that covers classical models of risk as well as new models of ambiguity.

We first show that  the absence of arbitrage and  the (properly defined) economic viability of the model are equivalent. In equilibrium, there are no arbitrage opportunities; conversely, for arbitrage--free asset pricing models,  it is possible to  construct a heterogeneous agent  economy   such that the   asset prices  are   equilibrium prices of that economy. 

The  second  key result is  the Fundamental Theorem of Asset Pricing. In contrast to risk, it is no longer possible to characterize 
viability  through the existence of a single 
linear pricing measure (or equivalent martingale measure).  
Instead, it is necessary to use a  suitable 
\emph{nonlinear} pricing  expectation,  that we call a sublinear martingale expectation. A sublinear expectation has the common properties of an expectation including  monotonicity, preservation of constants,  and positive homogeneity,  yet it  is no longer additive. Indeed, sublinear expectations can be   represented as the supremum of a class of (linear) expectations, an operation that does not preserve linearity\footnote{In economics, such a representation theorem appears first in \cite{GilboaSchmeidler89}. Sublinear expectations also arise in Robust Statistics, compare \cite{Huber81}, and they play a fundamental role in theory of risk measures in Finance, see \cite{Artzneretal99} and \cite{FoellmerSchied11}.}.
Nonlinear expectations arise   in decision--theoretic models of ambiguity--averse preferences (\cite{GilboaSchmeidler89}, \cite{Maccheronietal06}). It is interesting to see that a similar nonlinearity arises here for the \emph{pricing} functional. A general theory of equilibrium with such sublinear prices is developed in  \cite{BeissnerRiedel19}\footnote{Given that we have a nonlinear price system, one might wonder whether    agents can generate arbitrage gains  by splitting a consumption bundle into two or more plans. The convexity of our price functional excludes such arbitrage opportunities, see Proposition 1 in \cite{BeissnerRiedel19}.}.

The common order shapes equilibrium asset prices. We study various common orders and how they are related to versions of the Efficient Market Hypothesis (\cite{Fama70}). A strong  interpretation of the Efficient Market Hypothesis  says that properly discounted expected returns   of assets are equal to the return of a safe bond. We obtain this conclusion when the common order is based on expected payoffs with respect to a common prior.
When the common order is given by the almost sure order under a common prior,  one obtains a weaker version of the Efficient Market Hypothesis: under an equivalent pricing measure, expected returns are equal\footnote{This statement is equivalent to the classic  version of the Fundamental Theorem of Asset Pricing \cite{HarrisonKreps79, HarrisonPliska81, DuffieHuang85, DalangMortonWillinger90,DelbaenSchachermayer98}.}. In situations of Knightian uncertainty, different specifications of the common order can be made. An example is the quasi-sure order induced by a set of priors: a claim dominates quasi-surely another claim if it is almost surely greater or equal under all considered probability measures. Another example is the order induced by smooth ambiguity preferences, as introduced by \cite{klibanoff2005smooth}, where Knightian uncertainty is modeled by a second-order prior over a class of multiple priors. We show that weaker versions of the Efficient Market Hypothesis hold, depending on
the strength of the assumptions we are willing to impose on the common order, and 
how the related fundamental theorem of asset pricing needs to be suitably adapted.

\subsection*{Further Related Literature}

The relation of arbitrage and viability has been discussed in various contexts.  \cite{jouini1995martingales} and \cite{jouini1999viability} discuss models with transaction cost and other frictions.  \cite{Werner87} and \cite{Dana99} study the absence of arbitrage in its relation to equilibrium when a finite  set of agents is fixed a priori whereas  
\cite{Cassese2017} characterizes  the absence of arbitrage in an order-theoretic framework derived from coherent risk measures.  
Knightian uncertainty is also closely related to robustness concerns that play an important role   in macroeconomic models that deal with 
the fear of model misspecification (\cite{HansenSargent2001,HansenSargent08}). 
The pointwise order corresponds to the ``model-independent'' (or rather ``probability--free'') approach in finance that has been discussed, e.g., in \cite{Riedel14},  \cite{Acciaioetal}, \cite{Burzoni-et-al} and \cite{C}. This literature uses  different notions of ``relevant payoffs''  that our approach allows to unify under a common framework.

The paper is set up as follows. Section \ref{s.setup} describes the model,   the two main contributions of this paper, and provides four examples. Section \ref{s.examples} derives   various classic and new forms of the Efficient Market Hypothesis. The modeling philosophy  is discussed in more detail in   Section
\ref{s.discussion}. Section \ref{s.proofsmain} is devoted to the  proofs of the main theorems. 
The appendix contains a detailed study of  general discrete time markets when the space of contingent payoffs consists of bounded measurable functions. It  also discusses further extensions as, e.g., the equivalence of absence of arbitrage and absence of free lunches with vanishing risk, or the question if an optimal superhedge for a given claim exists.

\section{A Probability-Free Foundation for Financial Economics}
\label{s.setup}

A non-empty set $\Omega$ contains the states of the world; the $\sigma$--field $\cF$ on $\Omega$ collects the possible events.

The commodity space (of contingent claims) $\cH$ is a vector space of $\cF$-measurable real-valued functions containing all constant functions. We will use the symbol $c$ both for real numbers as well as for constant functions. $\cH$ is endowed with a metrizable topology $\tau$ and a pre-order $\le$ that are compatible with the vector space operations.  

The abstract vector space model allows to cover the typical models that have been used in financial economics. Under risk, it is common to take a space of suitably integrable  functions with respect to a given prior with the usual almost sure order; without an ex ante given probability measure, integrability cannot be used as a criterion. We thus allow for more generality here in order to include, e.g.,  spaces of suitably bounded measurable functions.

We assume throughout that agents' preferences are monotone with respect to the preorder $\le$ which thus  plays a crucial role in our analysis. 
A major conclusion of our study is that the strength of the assumptions we are willing to make on the common order (and therefore on the agents populating the economy)  shapes the results about market returns as we shall see in detail in Section \ref{s.examples}.
We assume throughout that the preorder $\le$ is consistent with the order on the reals for constant functions and with the pointwise order for measurable functions. 
A consumption plan $Z \in \cH$ is negligible if we have $0 \le Z $ and $Z \le 0$. $C \in \cH$ is nonnegative if $0 \le C$ and positive if in addition not $C \le 0$. We denote by $\cZ$, $\cP$ and $\cP^+$ the class of negligible, nonnegative and positive contingent claims, respectively.

We   introduce a class of  \emph{relevant contingent claims} $\cR$, a convex subset of $\cP^+$.  The relevant claims are used below in two important ways. On the one hand, they  describe the  nonnegative consumption plans that some agents strictly prefer to the null claim.
On the other hand, 
they signal arbitrage: if a net trade allows to obtain a payoff that dominates a relevant payoff with respect to the common order, we speak of an arbitrage. In the spirit of \cite{Arrow53} and most of the literature, a common choice of the relevant claims is the set of positive  claims $\cP^+$; we invite the reader to make this identification at first reading. However, 
it might be of interest to consider smaller relevant
sets in some economic contexts.  The introduction of $\cR$ also allows to subsume various notions of arbitrage that appeared in the literature, compare the discussion in Section \ref{s.discussion}.

The financial market is modeled by the set of \emph{net trades} $\cI\subset \cH$, 
a convex cone containing $0$. $\cI$ is   the set of payoffs that the agents 
can achieve from zero initial wealth by trading in the financial market.
In the basic frictionless model of  securities, $\cI$ contains the payoffs of self-financing strategies with zero initial capital.
In a frictionless market,  $\cI$ is a subspace. When   short selling constraints, credit line limitations, or transaction costs are imposed,  e.g.,   we are  led to a convex cone instead of a subspace, see Example \ref{ex.dynamic}.

An \emph{agent} 
in this economy is described by a    preference   relation $\preceq$
(i.e. a complete and  transitive binary relation) on $\cH$ that is 
\begin{itemize}
	\item \emph{weakly monotone with respect to $\le$}, i.e.  $X \oleq Y$ implies
	$X \preceq Y$ for every $X,Y \in \cH$;
	\item \emph{convex}, i.e. the upper contour sets $\{Z\in \cH: Z \succeq X\}$ are  convex;
	\item \emph{$\tau$-lower semi--continuous}, i.e. for every sequence $\{X_n\}_{n=1}^\infty \subset \cH$ 
	converging to $X$ in $\tau$ with $X_n \preceq Y$ for $n \in \mathbb N$, we have $X \preceq Y$.
\end{itemize}
The set of all agents is denoted by $\cA$. 

In the spirit of  \cite{HarrisonKreps79}, we think of a potentially large set of agents about whom some things are known, without assuming that we know exactly their  preferences or their number.  We only impose a list of properties on preferences that are standard in economics.  In particular, bearing in mind the interpretation of $\le$ as a \emph{common order}, the preferences are monotone with respect to $\le$.   Moreover, we impose some weak form of continuity with respect to the given topology $\tau$; it is known that, in general, some form of continuity is required for the existence of equilibrium. Convexity reflects a preference for diversification.
We stress that preferences are defined on the entire commodity space $\cH$; this assumption can be relaxed (we refer to Appendix \ref{s.eh} for the technical details).

A financial market
$(\cH,\tau,  \oleq,\cI,\cR)$ is
\emph{viable} if  there is a family of agents\\  $\{\preceq_a\}_{a\in A}\subset \cA$ 
such that

\begin{itemize}
	\item $0$ is optimal for each agent $a \in A$, i.e.
	\begin{equation}
		\label{EqnViable1}
		\forall    \ell \in \cI
		\quad \ell \preceq_a 0,
	\end{equation}
	\item for every relevant claim $R \in \cR$ there exists an agent $a \in A$ such that
	\begin{equation}
		\label{EqnViable2}
		0 \prec_{a}  R \,.
	\end{equation}
\end{itemize}
We say that $\{\preceq_a\}_{a\in A}$ 
{\em{supports the financial market}} $(\cH,\tau, \oleq,\cI,\cR)$.

A  market is in equilibrium when  agents have no incentive to trade away from their current endowment\footnote{We take the endowment to be zero  in our definition; this comes without loss of generality, see also the discussion in Section \ref{s.discussion}.}. 
We generalize the viability definition of \cite{HarrisonKreps79} who use a single representative agent to allow for  economies with  heterogeneous agents. Heterogeneity of agents allows us to prove the two  main theorems  in great generality with easy arguments without the need to construct strictly monotone preferences (that might not exist in all cases). We replace the strict monotonicity assumption by the weaker Condition (\ref{EqnViable2}) which, in particular, excludes the trivial case of agents who are indifferent between all payoffs.  Compare also our discussion of the viability concept in Section \ref{s.discussion} below.

A net trade $\ell \in \cI$ is an arbitrage if there exists a relevant claim
$R^* \in \cR$ such that $\ell \ge R^*$. More generally, 
a sequence of net trades
$\{\ell_{n}\}_{n=1}^\infty  \subset \cI$
is a \emph{free lunch with vanishing risk} 
if there exists a relevant claim
$R^* \in \cR$
and a   sequence $\{e_n\}_{n=1}^\infty \subset \cH$ of nonnegative consumption plans with $e_n\stackrel{\tau}{\rightarrow}0$
satisfying
$
e_n+ \ell_n \ge R^*$ for all $n \in\mathbb N$. We say that the financial market is 
{\em{strongly free of arbitrage}} if there is no free lunch with vanishing risk. In general, the absence of arbitrage is not equivalent to the absence of free lunches with vanishing risk\footnote{Our notion of free lunch with vanishing risk is a modification of the one used by \cite{Kreps81}. It corresponds to the notion proposed by  \cite{DelbaenSchachermayer98} for financial markets in continuous time in a probability space framework. This notion and our adapted viability concept allow to prove the equivalence of viability, absence of arbitrage, and the existence of a suitable pricing functional in contrast to \cite{Kreps81}.}. In Appendix \ref{s.versus}, we establish the equivalence for finite horizon discrete time financial markets.

\subsection*{Two Theorems}

Our first   theorem establishes the equivalence of viability and absence of arbitrage.

\begin{thm}
	\label{t.viable}
	A financial market is strongly free of arbitrage if and only if it is viable.
\end{thm}

In the  standard literature, the  model of the economy  is constructed on a   probability space with a given reference probability $\P$ that determines the null sets and the topology of the model. In such common prior models,   a financial market is viable if and only if there exists a linear pricing measure in the form of a  risk-neutral probability measure $\P^*$ that is equivalent to $\P$, as \cite{HarrisonKreps79} have shown. In the absence of a common prior, we have to work with a more general, sublinear notion  of  pricing.
A functional 
$$
\cE :  \cH \to\R \cup\{\infty\}
$$ is a {\emph{sublinear expectation}}
if  it is monotone with respect to $\oleq$,
translation-invariant, i.e. 
$
\cE(X+c)=\cE(X)+c
$
for all constant claims $c \in \cH$ and $X \in\cH$, 
and sublinear, i.e. for all $X,Y \in \cH$ and $\lambda>0$, we have 
$\cE(X+Y) \le \cE(X)+\cE(Y)$  and $\cE(\lambda X)=\lambda \cE(X)$.
$\cE$  has \emph{full support} if
$\cE(R) >0$ for every $R \in \cR$.
Last but not least, $\cE$ has the \emph{martingale property} if 
$\cE(\ell) \le  0$ for every
$\ell \in \cI$. We say in short that $\cE$ is a sublinear martingale expectation with full support if all the previous properties are satisfied.

It is well known from decision theory that sublinear expectations can be written as upper expectations over a set of probability measures. In our more abstract framework, probability measures are replaced by suitably normalized functionals. 
We say that $\varphi \in \cH^\prime_+${\footnote{$\cH^\prime$
		is the topological dual of $\cH$ and $\cH^\prime_+$ is the set of positive elements in $\cH^\prime$.}}
is a \emph{martingale functional}\footnote{In this generality the terminology \emph{functional} is more appropriate. When the dual space $\cH^\prime$ can be identified with a space of measures, we will use the terminology \emph{martingale measure}. The technical question whether these measures are countably additive is discussed in Appendix \ref{s.ca}.} if it satisfies $\varphi({ 1})=1$ (normalization) and $\varphi(\ell) \le 0$ for all $\ell \in \cI$.  In the spirit of the probabilistic language, we call a linear functional absolutely continuous if it assigns the value zero to all negligible claims. 
We denote by $\cQ_{ac}$
the set of absolutely continuous
martingale functionals.

The notions that we introduced now allow us to state the general version of the fundamental theorem of asset pricing in our order-theoretic context.  

\begin{thm}[Fundamental Theorem of Asset Pricing]
	\label{t.subexp}
	The financial market is viable if and only if there   exists a lower semi--continuous sublinear martingale expectation with full support. 
	
	In a viable market, the set of absolutely continuous martingale functionals $\cQ_{ac}$ is not empty and 
	$$ 
	\cE_{\cQ_{ac}}(X):=\sup_{\phi \in \cQ_{ac}} \phi(X)
	$$ 
	is a  lower semi--continuous sublinear martingale expectation with full support. Moreover, $\cE_{\cQ_{ac}} $ is maximal, in the sense that any other 
	lower semi--continuous sublinear martingale expectation with full support $\cE$ satisfies $\cE(X)\le \cE_{\cQ_{ac}}(X)$ for all $X\in\cH$.
\end{thm}

\begin{rem}
	Under nonlinear expectations, 
	one has to distinguish martingales from symmetric martingales; 
	a symmetric martingale has the property that the process itself 
	and its negative are martingales. In our context, this condition translates to the equality $\cE_{\cQ_{ac}}(\ell)=0$ for all $\ell \in \cI$.
	When the set of net trades $\cI$ is a linear space
	as in the case of frictionless markets,
	a net trade $\ell$ and its negative $-\ell$ belong to $\cI$. In this case, sublinearity and the condition 
	$\cE_{\cQ_{ac}}(\ell)\le 0$ for all $\ell\in\cI$ imply 
	$\cE_{\cQ_{ac}}(\ell)=0$ for all net trades $\ell\in\cI$. 
	Thus, the net trades $\ell$ are  symmetric  martingales.
\end{rem}

\subsection*{Four Examples}

Our   novel approach builds the foundations of  financial economics without imposing any a priori probabilistic structure, thereby including the new paradigm of Knightian uncertainty into the theory.  
We illustrate the unifying power of the model with four examples ranging from classical situations of risk to new ones with ambiguity.  


\begin{exm}\label{ex.atom}(The atom of finance and complete  markets)
	\label{ex.simple}
	The basic
	one--step binomial model, that we like to call the atom of finance, consists of two states of the world, 
	$\Om=\{1,2\}$. 
	An  element $X \in \cH$  
	can be identified with a vector in  $\R^2$.  Let $\oleq$ be the usual partial order of
	$\R^2$. The relevant claims are the positive ones, $\cR=\cP^+$.
	
	There is 
	a riskless asset $B$ and a risky asset $S$.
	At time zero, both assets have value $B_0=S_0=1$. 
	The riskless asset yields $B_1=1+r$ for an interest rate $r>-1$ at time 
	one,  whereas the risky asset takes the values $u$ in state 1 
	and respectively $d$ in state 2 with $u>d$. 
	
	We use the riskless asset $B$ as num\'eraire.  The discounted net return on the risky asset is
	$\hat \ell := S_1/(1+r)-1$.
	$\cI$ is the linear space spanned by $ \hat \ell$.
	There is no arbitrage if and only if the unique candidate 
	for a martingale probability of state one 
	$$
	p^*=\frac{1+r-d}{u-d}
	$$ belongs to $(0,1)$ which is equivalent to $u>1+r>d$. 
	$p^*$ induces the unique martingale measure $\P^*$ with expectation
	$$\E^*[X]= p^* X (1)
	+ (1-p^*) X_1(2)\,.$$  $\P^*$ is a linear measure;   moreover, it has the full support property since for every $R\in\cR$ we have  $\E^*[R]>0$.
	The market is viable with the single-agent economy  $A=\{\preceq^*\}$ where  the preference relation $\preceq^*$ is given by the linear expectation $\P^*$, i.e.  $X \preceq^* Y$ if and only if $\E^*[X] \le \E^*[Y]$.
	Indeed, under this preference $\ell \sim^* 0$ for 
	any $\ell \in \cI$ and $X\prec^* X+R$ for any $X \in \cH$ and
	$R \in \cP^+$.  In particular, any $\ell \in \cI$ is an optimal 
	portfolio and the market is viable.
	
	The preceding analysis carries over to all finite $\Omega$ 
	and complete financial markets. 
	
\end{exm}

\begin{exm}\label{ex.Ellsberg} (Ellsberg Market)
	We illustrate the concepts used in our definition of viability and in Theorem \ref{t.subexp} with the help of a market that is inspired by the
	Ellsberg thought experiments,   the archetypal instances  of ambiguity in economics and decision theory.
	In these experiments, an urn is called ambiguous if the  exact composition is not known to the participants.
	Suppose for simplicity that we model an ambiguous urn that contains  black and red balls; let it be known that the proportion of red balls is in the interval $[p_*,p^*]\subset [0,1]$.
	The finite state space is given by
	$
	\Omega:=\{ \mbox{red} , \mbox{black}\}.
	$
	The commodity space $\cH=\R^\Omega$ is the set
	of functions on $\Om$
	with the usual order.  	
	Consider the claim that
	pays one dollar if a red ball is drawn, i.e.
	$$
	X(\om)=
	\left\{ 
	\begin{array}{ll}
		1 \quad \ {\text{if}}\ \ &\om\ {\text{is red}},\\
		0 \quad \ {\text{if}}\ \ &\om\ {\text{is black}}
	\end{array}.
	\right.
	$$
	In contrast to the frictionless atom of finance, we now suppose that due to ambiguity, the 
	asset can be bought at  price $p_*$  and sold   at price $p^*>p_*$\footnote{Such a price setting behavior is quite natural in today's regulated financial markets when banks take the   \emph{model uncertainty} or ambiguity of their  algorithms into account and perform a number of stress tests to compute an interval of possible prices, compare the early account of such practices in \cite{Artzneretal99}.}.
	Then the set of net trades is  given by positive cone (rather than a linear subspace)  generated by 
	$
	\ell_1(\om)=X(\om)-p^*$ and 
	$\ell_{2}(\om)=p_* - X(\om), $ i.e.
	$$
	\cI
	=\left\{ \lambda \ell_1+ \mu \ell_2\ :\ \lambda, \mu \ge 0\ \right\}.
	$$
	Let $\P_p$ be the measure that assigns
	probability $p$ to the event $\{\mbox{red}\}$.
	The  risk-neutral subjective expected utility agent's preference  with this subjective belief $\P_p$ is denoted
	by $\preceq_p$.
	For this financial market, the set of absolutely continuous 
	martingale expectations $\cQ_{ac}$
	is given by 
	$$
	\cQ_{ac}=\{\P_p; p \in[ p_*,p^*]\}.
	$$
	The set of  agents $\cA^*=\{\preceq_p; p \in [p_*,p^*]\}$
	supports this market in the sense of our viability definition.
	The corresponding \emph{sublinear martingale expectation } is given by  
	$$
	\cE(\xi)= \max\{\ \E_{p_*}[\xi]\ ,\ \E_{p^*}[\xi]\ \}.
	$$

	In contrast to the frictionless, ambiguity-free, and complete atom of finance, we now have a continuum  of  heterogeneous risk-neutral agents $\cA^*$ that support the market (instead of a unique risk-neutral agent).
	The model's ambiguity (or, the imprecise probabilistic information) is   described  by the sublinear pricing functional  $\cE(\xi)$. The corresponding ambiguity-averse  Gilboa-Schmeidler agent's preferences are represented by the  utility function 	$$
	U_{GS}(\xi):= 
	\inf_{p \in [p_*,p^*]}\ \E_p[\xi].
	$$
	We have then 
	$
	U_{GS}(\ell) \le 0$ for all $ \ell \in \cI.
	$
	When $0<p_*<p^*<1$, $U_{GS}$ is   strictly monotone and 
	supports the  market.  	Note however, that when $p_*=0$ or $p^*=1$, the agent with utility function $U_{GS}$
	does no longer support the market as strict monotonicity fails. 
\end{exm}

\begin{exm} [Multiple priors]
	\label{ex.nullsets} 
	Knightian uncertainty leads to frameworks in which single ambiguity-averse or expected utility maximizers might not suffice to support a given arbitrage-free market\footnote{In a recent paper, \cite{Bartl19} analyzes portfolio and consumption choice in such discrete time models.}.  Consider a one-step financial market that allows trading
	only at time zero.  
	Suppose that the underlying financial market  returns are too ambiguous  or too  non-stationary to be accurately estimated or  modeled   by a single probability distribution, yet   agents are willing to take a stand on a range of possible models for returns.  Such a situation of  \emph{imprecise probabilistic information} about the distribution of asset returns can be modeled by  a 
	set of priors
	$\cQ$ on a measurable space $(\Omega,\cF)$.
	Let the common order be the
	\emph{quasi--sure} order induced by the family of priors, i.e.	
	$$
	X \oleq Y \quad
	\Leftrightarrow
	\quad
	\Q (X\le Y)=1, \ \ \forall\ \Q \in \cQ.
	$$
	Then, an event $A \subset \Om$ is  \emph{negligible}\footnote{These events are called  \emph{polar} in the mathematical literature on quasi-sure analysis.}
	if $\Q(A)=0$ for every prior $\Q\in \cQ$.	The commodity space that represents all claims in this market is given by
	$$
	\cH = \bigcap_{\Q\in \cQ} \L^2(\Om,\Q).
	$$
	As the $\L^2$ spaces are defined
	as equivalence classes, two claims in $\cH$ are in the same equivalence class
	if they differ only on a negligible set.
	Call a  claim $R \ge 0$   relevant if the set
	$\{R>0\}$ is not negligible, i.e., $\cR=\cP^+$.
	The set of net trades can be  any convex cone included in
	$$
	\cI_0=\{ \ell \in \cH \ :\ \E_\Q[\ell] \le 0 \quad \forall \ \Q \in \cQ\}.
	$$
	
	One might guess that the single Gilboa-Scheidler agent with the utility function
	$$
	U_{GS}(X):= \inf_{\Q\in \cQ} \E_\P[X], \qquad X \in \cH
	$$
	supports the financial market.  
	However, this agent would support the financial market
	only when it satisfies the monotonicity condition
	\eqref{EqnViable2}, i.e. $U_{GS}(R)>0$ for every $R\in\cR$. 
	This monotonicity implies that all priors must have the same null sets.
	Indeed, suppose that there are two priors $\P_1, \P_2 \in \cQ$ and
	an event $A$ satisfying $\P_1(A)>0=\P_2(A)$.  Then, $A$ is not negligible
	and consequently, the claim
	$R=\chi_A $ is relevant.  
	On the other hand, $U_{GS}(\chi_A) =0$ showing
	that the Gilboa-Schmeidler agent is not strictly monotone
	when there are priors with different null-sets\footnote{Note that the 
		same conclusion holds if we define $U_{GS}$ with 
		any set of probability measures equivalent to $\cQ$.}.
	
	Similarly, suppose that a single agent with a linear preference relation  $\preceq_\Q$ given by the subjective expected utility function $U(X)=\E_\Q [X]$ for  a prior $\Q$ supports the financial market\footnote{We take a risk-neutral agent for simplicity. The argument holds as well for risk-averse agents with a standard Bernoulli utility function.}.
	Let $A$ be a null set of $\Q$.   Then, this agent is   indifferent
	between $\chi_A$ and the zero claim.  Thus  the monotonicity 
	condition \eqref{EqnViable2} implies that 
	$\chi_A$ cannot be relevant and $A$ must be negligible.
	As   negligible sets are null for every prior, we conclude that
	the null sets of $\Q$ and the negligible  sets coincide.
	
	We conclude  that   we need to reconsider  the viability concept in models that do not allow to describe the negligible sets by a single probability measure.  In particular, note that the family of heterogeneous
	agents $\{\preceq_\Q\}_{\Q\in \cQ}$
	with linear preferences induced by the set of priors $\cQ$   
	supports the financial market.
\end{exm}

\begin{exm} [Volatility Uncertainty]
	\label{ex.G} 
	In the classical Samuelson-Merton-Black-Scholes model, the stock price  is a geometric Brownian motion satisfying the stochastic differential
	equation $dS_t= \sigma S_t dB_t$ where
	$B$ is a standard Brownian motion and $\sigma$ is
	a positive constant modeling the {\emph{volatility}}.  
	
	It is widely recognized and empirically well documented  that volatility is time-varying and stochastic; a variety of complex stochastic models of the underlying dynamics have been proposed in turn\footnote{Starting with the Heston model (\cite{Heston93}), a whole literature has explored ever more complex dynamics, see \cite{Ghysels96} for an overview.}. One might question whether it is plausible to assume that agents  can identify one of these models uniquely, in particular when the relevant volatility is not directly observable. A robust modeling approach thus allows for a whole class of volatility models.
	
	Let us  consider a financial market with Knightian 
	uncertainty about the  volatility of the price process, as discussed in  \cite{EpsteinJi13} 
	for con\-sumption-based asset pricing and in  \cite{Vorbrink14} for option pricing;   
	\emph{any adapted process} $\sigma$ taking values in a certain interval $[\underline{\sigma},\overline{\sigma}]$ 
	is a plausible volatility process. We denote this class by $\Sigma$;  for a given $\sigma\in\Sigma$, we let $\P^\sigma$ be
	the corresponding distribution of the stock price process.  The ambiguity about volatility is thus modeled by the 
	the set of priors $\cM= \{\P^\sigma : \sigma \in \Sigma\}$. 
	Take  $\cR=\cP^+$, and model the financial market by taking  $\cI$ to be the set of all
	stochastic integrals with simple
	integrands that are bounded from below{\footnote{We refer the 
			reader to \cite{soner2012wellposedness,soner2013dual}
			for all the technical details, in particular, 
			for the formal construction of $\{\P^\sigma : \sigma \in \Sigma\}$ and its subtle properties.}}.
	
	In this situation, priors in $\cM$ are not mutually equivalent; 
	in fact, typically they are singular to each other.
	A single probability space framework is  not able to capture volatility uncertainty. 
	In most finance models, the assumption of a reference probability was made for technical reasons, in order to be able to apply the Girsanov theorem. The restriction to a single probability space framework in diffusion models  limits the set of possible models; only drift uncertainty can be captured by such models, not volatility uncertainty. The issue becomes  even more pertinent if one wants to capture more complex financial models involving jumps.

	As the   negligible sets  of $\cM$ 
	cannot be generated  by any single probability measure,
	this financial market is not covered by the classical probability space framework.  However, the family of heterogenous agents 
	$\{\preceq_\sigma\}_{\sigma \in \Sigma}$ induced by the 
	priors $\P^\sigma$ does support the financial  market. 	Indeed, for every prior $\P^\sigma$,  every $\ell \in \cI$  is a $\P^\sigma$-local martingale
	and consequently, $\E_ {\P^\sigma}[\ell] \le 0$.

	We now discuss briefly a variation of the example in which the set of relevant claims is derived from preferences.   In a first step, let us consider agents who do not have a razor-sharp model of volatility, but are willing to use a robust version that is called ``$\epsilon$-contamination'' in statistics (\cite{Huber65}) and  decision theory (\cite{Eichberger99})
	For any reference volatility $ \sigma \in \Sigma$ and contamination $\eps>0$, let
	$$
	\Sigma_{\sigma,\eps}:= \{ \tilde{\sigma} \in \Sigma\ :\
	\|\sigma-\tilde{\sigma}\|_\infty \le \eps\ \}
	$$
	describe the sets of volatility models that are close to $\sigma$. 
	The $\epsilon$-contamination preferences  $\preceq_{\sigma,\eps}$ are represented by the utility functions
	$$
	U_{\sigma,\eps}(X):= \inf_{\sigma \in \Sigma_{\sigma,\eps}} \E_{\P^\sigma}[X] .
	$$
	We use these preferences to define the set of relevant claims as follows by setting 
	$$
	\cR=\{R \in\cP_+:   U_{\sigma,\eps}(R)>0   \text{ for some}
	\ \ \sigma \in \Sigma, \eps>0 \}.
	$$
	A payoff on an event  is thus relevant if a bet on that event matters for some ambiguity-averse agent with $\epsilon$-contamination preferences.
	This technically more complex model is also included in our framework. 
	Indeed, for  $\sigma \in \Sigma$ and $\ell \in \cI$, 
	$\E_{\P^\sigma}[\ell] \le 0$  and thus 
	$U_{\sigma, \eps}(\ell)\le 0$, 
	proving the optimality  condition (\ref{EqnViable1}). 
	The  definition of $\cR$ shows directly that
	the monotonicity condition \eqref{EqnViable2} is satisfied by 
	the family of heterogenous agents
	$\{\preceq_{\sigma,\eps}\}_{\{\sigma \in \Sigma, \eps >0\}}$.
	
\end{exm}


\section{The Efficient Market Hypothesis}
\label{s.examples}

The Efficient Market Hypothesis (EMH) plays a fundamental role in the history of Financial Economics.  \cite{Fama70} calls  markets   informationally efficient if all available information is reflected properly in current asset prices. There are several interpretations of this conjecture; in the early days after its appearance, the efficient market hypothesis was usually interpreted as asset prices being random walks in the sense that (log-) returns be independent from the past and identically distributed with the mean return being equal to the return of a safe bond (\cite{Malkiel2003}). Later, the informational efficiency of asset prices was interpreted as a  martingale property; (conditional) expected returns of all assets are equal to  the return of a safe bond under \emph{some} probability measure.   This conjecture of the financial market's being a ``fair game''   dates back to  \cite{Bachelier1900} and was  rediscovered by Paul Samuelson (\cite{Samuelson1965,Samuelson1973}). In dynamic settings, market efficiency is thus  strongly related to (publicly available)  information.  Under Knightian uncertainty, the role of information and the martingale property of prices needs to be adapted properly as we shall see in this section\footnote{We refer to \cite{JL12} for a detailed  analysis of the interplay between different information sets and market efficiency under a common prior. In our  framework, the information flow is taken as given; it is implicitly encoded in the set of available claims $\cI$. We do not consider the issue of private information of insiders. }.


Throughout this section, let
us assume that we have a frictionless one--period or discrete--time multiple period financial market as in Example \ref{ex.dynamic}. In particular, the set of net trades $\cI$ is a subspace of $\cH$.


\subsection{A Strong Version of the  Efficient Market Hypothesis under Risk}
\label{ex.strong}

There are various ways to formalize the Efficient Market Hypothesis. A particular strong interpretation of informational efficiency requires that expected returns of all risky investment be equal under  a common prior. In our framework, such a conclusion results if the common order is 
derived from a common prior.

Let  $\P$ be a probability measure
on $(\Om,\cF)$.  Set $\cH=  \cL^1(\Om,\cF,\P)$.  
Let the common order be defined by saying   $X\le Y$ if  and only if the expected payoffs under  $\P$ satisfy
\begin{equation}
	\label{EqnDefSEMH}
	\E_\P[X] \le \E_\P[Y].
\end{equation}
We call $\P$ the \emph{common prior} of this model.
In this case, negligible claims coincide with the claims 
with mean zero under $\P$.  Moreover, $X \in \cP$
if $\E_\P[X]\ge 0$. We take $\cR=\cP_+$.

\begin{prop}
	Under the assumptions of this subsection, the  financial market is viable if and only if the common prior $\P$ is a martingale measure. In this case, $\P$ is the unique martingale measure.
\end{prop}

\begin{proof}
	Note that the common order as given by (\ref{EqnDefSEMH}) is complete. 
	If $\P$ is a martingale measure, the common order $\le$ itself defines a linear preference relation under which the market is viable with $A=\{\le\}$. 
	
	On the other hand if the market is viable, Theorem \ref{t.subexp} ensure that there exists a sublinear martingale expectation with full support.
	By the Riesz duality theorem, a martingale functional $\phi \in \cQ_{ac}$ can be identified with a probability measure $\Q$ on $(\Omega,\cF)$. It is absolutely continuous (in our sense defined above) if and only if it assigns the value $0$ to all negligible claims. As a consequence, we have 
	$E_\Q [X]=0$ whenever $E_\P[X]=0$. Then  
	$\Q=\P$ follows\footnote{If $\Q\not=\P$, 
		there is an event $A\in\cF$ with $\Q(A)<\P(A)$. 
		Set $X=1_A-\P(A)$. Then $0=E_\P[X]>\Q(A)-\P(A)=E_\Q[X]$.}.
\end{proof}

The only absolutely continuous 
martingale measure is
the common prior itself.  As a consequence, all traded assets have zero  net expected return  under the common prior. A financial market is thus viable if and only if the strong form of the expectations hypothesis holds true.

\subsection{A Weak Version of the  Efficient Market Hypothesis under Risk}
\label{ex.weakemh}
A weaker version of  the efficient market hypothesis   
states that expected returns    
are equal under some (pricing) probability measure $\P^*$ that is equivalent to the common prior (or ``real world'' probability) $\P$. 

Let $\P$ be  a probability on $(\Om,\cF)$  and $\cH=  \cL^1(\Om,\cF,\P)$. 
In this example, the common order is given by the  almost sure order  under the common prior  $\P$, i.e.,
$$
X \le Y \quad \Leftrightarrow
\quad
\P(X \le Y) =1.
$$
A payoff is negligible if it vanishes $\P$--almost surely and is positive if it is $\P$--almost surely nonnegative. 
Let the  relevant claims $\cR$ consist of  the $\P$--almost surely nonnegative payoffs that are strictly positive with positive $\P$--probability, 
$$
\cR= \left\{R\in \cL^1(\Om,\cF,\P)_+\ :\
\P(R>0) >0\right\}.
$$

A functional $\phi\in \cH'_+$ is an absolutely continuous martingale functional if and only if it can be identified with a probability measure $\Q$ that is absolutely continuous with respect to  $\P$  and if all net trades have expectation zero under $\phi$. In other words, discounted asset prices are $\Q$-martingales.
We thus obtain a version  of the Fundamental Theorem of Asset Pricing under risk, similar to 
\cite{HarrisonKreps79} and 
\cite{DalangMortonWillinger90}.

\begin{prop}\label{PropEMHRisk}
	Under the assumptions of this subsection, the  financial market is viable if and only if there is a martingale measure  $\Q$ that has a bounded density with respect to $\P$.
\end{prop}

\begin{proof}
	If $\Q$ is a martingale measure equivalent to $\P$, define $X \preceq^* Y$ if and only if $\E_{\Q}[X] \le \E_{\Q}[Y]$. Then the market is viable with $A=\{\preceq^*\}$.  Condition (\ref{EqnViable2}) is satisfied because $\Q$ is equivalent to $\P$. 
	
	If the market is viable, Theorem \ref{t.subexp} ensures that there exists a sublinear martingale expectation with full support.
	By the Riesz duality theorem, a martingale functional $\phi \in \cQ_{ac}$ can be identified with a probability measure $\Q_{\phi}$ that is absolutely continuous with respect to  $\P$, has a bounded density with respect to $\P$,  and   all net trades have zero expectation zero under $\Q_{\phi}$. In other words, discounted asset prices are $\Q_{\phi}$-martingales. From the full support property, the family $\{\Q_\phi\}_{\phi\in\cQ_{ac}}$ is equivalent to $\P$, meaning that $\Q_\phi(A)=0$ for every $\phi\in\cQ_{ac}$ if and only if  $\P(A)=0$. By the Halmos-Savage Theorem (\cite{HS49}, \cite[Theorem 1.61]{FoellmerSchied11}), there exists	a countable subfamily $\{\Q_{\phi_n}\}_{n\in\N}\subset\{\Q_\phi\}_{\phi\in\cQ_{ac}}$ which is equivalent to $\P$.   The measure $\Q:=\sum_{n=1}^{\infty}2^{-n}\Q_{\phi_n}$ is the desired equivalent martingale measure.
\end{proof}

\subsection{The Efficient Market Hypothesis under Knightian Uncertainty}

We turn our attention to the  EMH  under  Knightian uncertainty. We consider first the case when the common order is derived from a common set of priors, inspired by the multiple prior approach in decision theory (\cite{Bewley02,GilboaSchmeidler89}). We then discuss a second-order Bayesian approach that is inspired by the smooth ambiguity model (\cite{klibanoff2005smooth}).

\subsubsection{A Strong Version   under Knightian Uncertainty}
\label{ss.semhku}

We  consider a generalization of the original EMH to Knightian uncertainty  that shares a certain analogy with Bewley's incomplete expected utility model (\cite{Bewley02}) and Gilboa and Schmeidler's maxmin expected utility (\cite{GilboaSchmeidler89})\footnote{For the relation between the two approaches, compare also   the discussion of objective and subjective ambiguity in \cite{GilboaMaccheroniMarinacciSchmeidler10}.}.   

Let $\Om$  be a metric space with metric $d$ and Borel sets $\cF$. Let  $\cM$ be a convex, weak$^*$-closed set of   priors  
on $(\Om, \cF)$.  Define the semi-norm 
$$
\|X\|_\cM:= \sup_{\P\in \cM} \E_\P|X|.
$$
Let $\cL^1(\Om,\cF,\cM)$
be the closure of continuous and bounded functions on $\Om$
under the semi-norm $\|\cdot\|_\cM$.  If we identify the
functions which are $\P$-almost surely equal for every $\P \in \cM$,
then $\cH=\cL^1(\Om,\cF,\cM)$ is a Banach space.
The topological dual of $\cL^1(\Om,\cF,\cM)$
can be identified with  probability measures that admit a bounded density  with respect to some measure in $\cM$ (\cite{BionNadalKervarec12,BeissnerDenis18}).  Therefore, 
any absolutely continuous martingale functional $\Q\in \cQ_{ac}$
is a probability measure and
$\cM$ is closed in the weak$^*$ topology induced by $\cL^1(\Om,\cF,\cM)$.

Consider the uniform order induced by expectations over 
$\cM$, 
$$
X\oleq Y \quad
\Leftrightarrow
\quad
\forall \P\in\cM \quad \E_\P[X] \le \E_\P [Y]\,.
$$
Then, $Z \in \cZ$ if $\E_\P[Z]=0$
for every $\P \in \cM$.  A claim $X$
is nonnegative if $\E_\P[X]\ge 0$
for every $\P \in \cM$.
Let  the relevant claims   
consist of nonnegative claims with a positive return under some prior belief, i.e.
$$ 
\cR=\{
R \in \cH : 0 \le \inf_{\P\in\cM}\ \E_\P[R] \
{\mbox{ and }}
0<  \sup_{\P\in\cM}\ \E_\P[R]\}\,.
$$

\begin{prop}
	\label{p.431}
	Under the assumptions of this subsection, if the  financial market is viable, then
	the set of absolutely continuous martingale functionals $\cQ_{ac}$ is a subset of
	the set of priors $\cM$.
\end{prop}

\begin{proof}
	Set $\cE_\cM(X):= \sup_{\P \in \cM} \E_\P[X]$.
	Then, 
	$Y \le 0$ if and only if 
	$\cE_\cM(Y) \le 0$.
	Fix $\Q \in \cQ_{ac}$
	with the preference relation given by
	$X \preceq_\Q Y$ if $\E_\Q[X-Y] \le 0$.

	Let us assume that $\Q \notin \cM$.
	Since $\cM$ is a weak$^*$-closed and convex
	subset of the topological dual of $\cL^1(\Om,\cF,\cM)$,
	there exists 
	$X^* \in \cL^1(\Om,\cF,\cM)$
	with
	$\cE_\cM(X^*) <0 <\E_{\Q}[X^*]$ by the Hahn-Banach theorem.
	In particular, $X^* \in \cL^1(\Om,\cF,\cM)$ and 
	$X^* \le 0$.
	Since $\preceq_\Q$
	is weakly monotone with respect to $\le$,
	$X^* \preceq_\Q 0$.  Hence,
	$\E_{\Q}[X^*] \le 0$
	contradicting the choice of $X^*$.
	Therefore, $\cQ_{ac} \subset \cM$.

\end{proof}

Expected returns of traded securities are thus not necessarily  the same under all $\P\in \cM$. However, 
the set of martingale measures   is a subset of  $\cM$ here,  and  thus the strong form of the EMH
holds true on a  subset of the set of priors $\cM$.

In general, it is not possible to characterize the set of martingale measures in more detail. However, we can identify a subspace of claims on which expectations under all priors coincide.  Let $\cH_\cM$ be the subspace of claims
that have  no ambiguity in the mean in the sense that
$\E_\P[X]$ is the same constant for all $\P \in \cM$.
Consider the submarket $(\cH_\cM,\tau, \oleq,\cI_\cM,\cR_\cM)$
with $\cI_\cM:=\cI\cap\cH_\cM$ and $\cR_\cM:=\cR\cap \cH_\cM$.
Restricted to this market, the sets of measures $\cQ_{ac}$ and $\cM$ are identical
and the  strong EMH holds true.

The following simple
example illustrates these points.

\begin{exm}
	\label{ex.simple1}
	Let $\Om=\{0,1\}^2$, $\cH$ be all functions on $\Om$.  Then, $\cH=\R^4$ and
	we write $X=(x,y,v,w)$ for any $X\in \cH$. Let $\cI=\{(x,y,0,0): x+y=0\}$.
	Consider the priors given by
	$$
	\cM:=\left\{ \left(p, \frac12 -p, \frac14, \frac 14\right)\ : \ p \in \left[\frac16,\frac13\right]\right\}.
	$$
	There is Knightian uncertainty about the first two states, yet no Knightian uncertainty about the last two states.

	A claim $Z=(Z_1,Z_2,Z_3,Z_4)$ is negligible iff $Z_1=Z_2$ and $Z_4=-2Z_1-Z_3$. In particular, $X=(1,1,0,-2)$ and $Y=(0,0,1,-1)$ are negligible. 
	Now let $\Q^*=\left(q_1,q_2,q_3,q_4\right) \in \cQ_{ac}$. The martingale property implies $q_1=q_2$. Absolute continuity requires that $\E_{\Q^*}[X]=0$ and  $\E_{\Q^*}[Y]=0$, or
	$q_1+q_2-2q_4=0$ and $q_3=q_4$. From here, we obtain with $q_1=q_2$ that $q_1=q_2=q_3=q_4$, so
	$\Q^*=\left(\frac14,\frac14,\frac14,\frac14 \right)$.
	Notice that $\Q^* \in \cM$.
	
	In this case,
	$\cH_\cM =\left\{ X=(x,y,v,w) \in \cH\ :\ x=y\right\}$.  
	In particular, all priors in $\cM$ coincide with $\Q^*$ when restricted to $\cH_\cM$.  Hence, for the claims that are mean-ambiguity-free,
	the strong efficient market hypothesis
	holds true. 
\end{exm}

\subsubsection{A Weak Version   under Knightian Uncertainty}
\label{e.choquet-utility}
Let $\cM$ be a common  set of  priors  
on $(\Om, \cF)$.  Let
$
\cH$ be the space of bounded, measurable functions.
Let the common order be given by the  quasi-sure ordering under the common set of priors $\cM$, i.e.
$$
X \oleq Y \quad
\Leftrightarrow
\quad
\P(X\le Y)=1, \ \ \forall\ \P \in \cM.
$$
In this case, a claim $X$ is negligible if it vanishes \emph{$\cM$--quasi--surely}, i.e. with probability one for all $\P\in\cM$. 
An indicator function $1_A$ is thus negligible if the set $A$ is  \emph{polar}, i.e. a null set with respect to every probability in $\cM$. 
Take the set of relevant claims to be\footnote{These sets of positive and relevant claims can be derived from Gilboa--Schmeidler utilities. Define 
	$
	X\preceq Y $ if and only if 
	$
	\inf_{\P\in\cM} E_\P [U(X)]  \leq  \inf_{\P\in\cM}  U(Y)]$   for all strictly increasing
	and concave real functions $U$.
	Then  $0\preceq Y$  implies that for all priors $\P \in \cM$, every risk-averse expected utility agent prefers $Y$ to the null claim. As is well known, this is equivalent to $Y$ dominating the null claim in the sense of second order stochastic dominance  for every $\P \in \cM$, implying that $Y$ is nonnegative almost surely for all $\P\in\cM$.}
$$
\cR=\left\{ R \in \cP\ :\ \exists \ \P \in \cM \quad 
\text{such that} \quad
\P(R>0)>0\ \right\}.
$$

\begin{prop}
	Under the assumptions of this subsection, the  financial market is viable if and only if there is a set of finitely additive martingale measures  $\cQ$  that has the same polar sets as the common set of priors $\cM$.
\end{prop}

\begin{proof}
	Suppose that the market is viable. We show that the class $\cQ_{ac}$ from Theorem \ref{t.subexp} satisfies the desired properties. The martingale property follows by definition and from the fact that $\cI$ is a linear space. Suppose that $A$ is polar. Then, $1_A$ is negligible and from the absolute continuity property, it follows $\phi(A)=0$ for any $\phi\in\cQ_{ac}$. On the other hand, if $A$ is not polar, $1_A\in\cR$ and from the full support property, it follows that there exists $\phi_A\in\cQ_{ac}$ such that $\phi_A(A)>0$. Thus, $A$ is not $\cQ_{ac}$-polar. We conclude that $\cM$ and $\cQ_{ac}$ share the same polar sets.
	For the converse implication, define $\cE(\cdot):=\sup_{\phi\in\cQ}\E_{\phi}[\cdot]$. Using the same argument as above, $\cE$ is a sublinear martingale expectation with full support. From Theorem \ref{t.subexp} the market is viable. \end{proof}

Under Knightian uncertainty,  there can be  indeterminacy in arbitrage--free prices as there is  frequently a range of economically justifiable arbitrage--free prices. Such indeterminacy has been observed in full general equilibrium analysis as well (\cite{RiSh05,DanaRiedel13,BeissnerRiedel19}). In this sense,  Knightian uncertainty shares a similarity with incomplete markets and other frictions like transaction costs, but the economic reason for the indeterminacy is different. 

\subsubsection{A second-order Bayesian version of the Efficient Market Hypothesis}  
\label{e.smooth}

We now consider a common order $\oleq$ obtained by a second--order Bayesian approach,   in the spirit of  the smooth ambiguity model (\cite{klibanoff2005smooth}).

Let $\cF$ be a sigma algebra on $\Om$
and $\mathfrak{P}=\mathfrak{P}(\Omega)$ the set of 
all probability measures on $(\Omega,\cF)$.  
Let $\mu$ be  a second order prior, i.e. a probability measure\footnote{From Theorem 15.18 of \cite{aliprantis1999infinite}, the space of probability measure is a Borel space if and only if $\Omega$ is a Borel space. This allows to define second order priors.} on $\mathfrak{P}$. 
The common prior in this setting is given by  the probability measure  $\hat{\P}:\cF\to[0,1]$ defined as $\hat{\P}(A)=\int_\mathfrak{P} \P(A) \mu(d\P)$. 
Let $\cH=  \cL^1(\Om,\cF,\hat{\P})$.

The common order is given by
$$
X \oleq Y \quad \Leftrightarrow
\quad
\mu\left(\{\P\in\mathfrak{P}\ :\ \P(X \le Y)=1\ \}\right)=1\,.
$$
A claim is positive if it is $\P$--almost surely nonnegative for all priors in the support of the second order prior $\mu$. 
A claim is relevant if   the set of beliefs $\P$ under which the claim is strictly positive with positive probability is not negligible according to the second order prior\footnote{The order used in this section   can be derived from smooth ambiguity utility functions. 
	Define 
	$
	X\le Y $ if and only if 
	$$
	\int_{\mathfrak{P}} \psi\left(\E_\P[U(X)]\right) \mu(d\P)  \le\int_{\mathfrak{P}} \psi\left(\E_\P[U(X)]\right) \mu(d\P)  $$
	for all strictly increasing
	and concave real functions $U$ and $\psi$. Recall that $\psi$ reflects uncertainty aversion.
	Then  $0\le Y$  is equivalent to $Y$ dominating the zero claim in the sense of second order stochastic dominance for $\mu$--almost all $\P\in\mathfrak{P}$.}.

\begin{prop}
	Under the assumptions of this subsection, the  financial market is viable if and only if there is a martingale measure  $\Q$ that has   the form
	$$
	\Q(A) = \int_\mathfrak{P} \int_A \,D   \,d\P   \,\mu(d \P)
	$$
	for some  state price density $D$.
\end{prop}

\begin{proof}
	The set function $\hat{\P}:\cF\to[0,1]$ defined as $\hat{\P}(A)=\int_\mathfrak{P} \P(A) \mu(d\P)$ is a probability measure on $(\Omega,\cF)$. The induced $\hat{\P}$-a.s. order coincides with $\oleq$ of this subsection. The result thus follows from Proposition \ref{PropEMHRisk} and the rules of integration with respect to $\hat{\P}$. \end{proof}

The smooth ambiguity model thus leads to a second--order Bayesian approach for asset returns. All asset returns are equal to the safe return for some second order martingale measure; the expectation is the average expected return corresponding to a risk--neutral second order prior $\Q$.

\section{Further   Discussion of the Modeling Philosophy}
\label{s.discussion}

This section discusses various aspects of our modeling approach in more detail. We first discuss the motivation of using a common order instead of a probabilistic framework. We then explain how various static and dynamic financial markets are  embedded in our abstract model. Our notion of viability is related to the usual notion of equilibrium of a competitive market with heterogeneous agents. We discuss the role of sublinear versus linear prices in our theory, and lastly show how our concept of relevant claims can be used to unify various notions of arbitrage in the finance literature.

\paragraph{A common order versus a common probabilistic framework }

Preference properties  shared by all agents in the market will be reflected in equilibrium prices.  

A situation of \emph{risk} is described by the fact that agents share a common prior; in the laboratory, a   random experiment based on an objective device like a roulette wheel or a coin toss simulates such a market environment.  If no such objective device can be invoked, as in the real world, one might still presume the existence of a common subjective belief  for all market participants as it is done (implicitly) in the Capital Asset Pricing Model as well as in its  its consumption-based version. Such an assumption might be too strong;  the  Ellsberg experiments show how  to create an environment of Knightian uncertainty in the lab.  In complex  financial markets in which  credit risk claims, options on  term structure shapes and   volatility dynamics are traded, Knightian uncertainty plays a prominent role  because agents lack   precise probability estimates of crucial model parameters and they might be wary of  potential structural breaks in the data. Moreover, as \cite{EpsteinJi13} have shown (see also Example \ref{ex.G}), if we model Knightian uncertainty about volatility, it is logically impossible to construct a reference probability measure. 

We take these considerations as a \emph{motivation} to  forego any explicit or implicit probabilistic assumption, be it   a common prior $\P$ or the weaker assumption of a common reference probability.  Instead, we  base our analysis on  a  \emph{common    order} $\le$, a far weaker assumption that only requires a  (typically incomplete) unanimous dominance criterion.  A minimal example of a common order is the pointwise order.  Pointwise dominance is certainly a criterion that we might assume to be unanimously shared in the context of monetary or single good payoffs. 
The generality of our approach allows to cover a wide variety of situations, including the well-studied case of risk as well as situations of Knightian uncertainty as we have seen in the previous section\footnote{It might be interesting to note an alternative way of setting up the model in which the common order is \emph{derived} from a class of given preference relations.  Suppose that no common order is a priori given.  Instead, we start with  a class of preference relations  $\cA_0$ on the commodity space $\cH$ that are  convex and $\tau$-lower semicontinuous.
	We can then define the \emph{uniform} order derived from the set of preference relations $ \cA_0$ as follows.
	Let  
	$$ 
	\cZ_\ppeq  := \left\{ Z \in \cH : 
	X\ppeq Z+X \ppeq X, \ \forall\  X \in \cH \right\},
	$$ 
	be the set of negligible (or null) claims for the preference relation $\ppeq \in {\cA_0}$. 
	We call $\cZ_{uni} := \bigcap_{\ppeq \in  \cA_0} \cZ_\ppeq$  the set of 
	\emph{unanimously negligible} claims. 
	Let  the uniform pre-order $\le_{uni}$ on $\cH$ be given by 
	$X \le_{uni} Y$ if and only if there 
	exists $Z \in \cZ_{uni}$ such that
	$X(\omega) \le  Y(\omega)+ Z(\omega)$ for all $\omega \in \Omega$. Note that we use the pointwise order on the reals and the set of  uniformly negligible payoffs to derive the common order from  $\cA_0$. 
	$(\cH,\le_{uni})$ is then a pre-ordered vector space, and agents' preferences in $\cA_0$ are monotone with respect to the uniform pre-order. 
	Let  $\cA$ be the set of (convex etc.) preferences  that  are monotone with respect to the common order. Note that $\cA$ contains $\cA_0$, but is usually larger than $\cA_0$.  }.

\paragraph{The Financial Market}

We model the financial market in a rather reduced form with the help of the convex cone $\cI$. 
This abstract approach is sufficient for our purpose of discussing the relation of arbitrage and viability. 
In the next example, we show how the usual models of static and dynamic trading are embedded.

\begin{exm}
	\label{ex.dynamic}
	We consider four markets with increasing complexity.
	
	1. In a one period setting with finitely many states $\Omega =\{1,\ldots, N\}$, a financial market with $J+1$ securities can be described by its initial prices $x_j \ge 0, j=0, \ldots, J$ and a $(J+1 ) \times N$--payoff matrix $F$, compare \cite{LeRoyWerner2014}.  A portfolio  $\bar{H}=(H_0, \ldots, H_J) \in \mathbb R^{J+1}$  has the payoff   $\bar{H} F=\left(\sum_{j=0}^J H_j F_{j\omega}\right)_{\omega=1,\ldots,N}$; its  initial cost satisfies $H \cdot x =\sum_{j=0}^J H_j x_j $. If the zeroth  asset is riskless with a price  $x_0=1$ and pays off $1$ in all states of the world, then a net trade with zero initial cost can be expressed in terms of the portfolio of risky assets $ H=(H_1, \ldots, H_J) \in \mathbb R^J$ and the return matrix $\hat{F}=(F_{j\omega}-x_j)_{j=1,\ldots,J,\omega=1,\ldots,N}$. 
	$\cI$ is given  by the image of the $J \times N$ return matrix $\hat{F}$, i.e.
	$$
	\cI=\{ H \hat{F}: H \in \mathbb R^J\}\,.
	$$
	
	2.  Our model includes the case of finitely many trading periods. Let 
	$\F:=(\cF_t)_{t=0}^T$ be a filtration on  $(\Om,\cF)$
	and  $S=(S_t)_{t=0}^T$ be an adapted
	stochastic process
	with values in $\R^J_+$ for some $J\ge 1$; $S$ models the uncertain assets. We assume that a riskless bond with interest rate zero is also given. Then, the set of net trades can be described by the gains from trade processes: 
	$\ell \in \cH$
	is in $\cI$ 
	provided that there exists predictable integrands
	$H_t\in \left(\cL^0(\Omega,\cF_{t-1})\right)^J$ for  
	$t=1,\ldots,T$ such that,
	$$
	\ell= \left(H \cdot S\right)_T:= \sum_{t=1}^{T} H_t\cdot \Delta S_t, \quad
	{\text{where}}
	\quad
	\Delta S_t:=(S_{t}-S_{t-1}).
	$$
	In the frictionless case, the set of net trades is  a subspace of $\cH$. In general, one might impose restrictions on the set of admissible trading strategies. 
	For example, one might exclude  short-selling of  risky assets, or impose a   bound on agents' credit line; in these cases, the marketed subspace $\cI$ is a convex cone, compare \cite{Luttmer96}, \cite{jouini1995martingales}, and \cite{araujo2018financial}, e.g.
	
	3. In   \cite{HarrisonKreps79}, the market is described by a marketed space 
	$M \subset \cL^2(\Omega, \cF, \P)$ and a (continuous) linear functional $\pi$ on $M$. In this case,     $\cI$ is the kernel of the price system, i.e.  $$\cI=\{ X \in M : \pi(X)=0\}\,.
	$$
	
	4. In continuous time, the set of net trades consists of stochastic integrals of the form
	$$
	\cI= \left\{ \int_0^T \theta_u \cdot dS_u\ : 
	\ \theta \in \cA_{adm}\ \right\},
	$$
	for 
	a suitable set of  {\emph{admissible}}  strategies  $\cA_{adm}$.  There are several possible choices of such a set.  
	When the stock price process $S$ is a semi-martingale one example of $\cA_{adm}$
	is the set of all $S$-integrable,
	predictable processes whose integral is bounded
	from below.{\footnote{In continuous time, 
			to avoid doubling strategies
			a lower 
			bound 
			(maybe more general than above)
			has to be imposed on the stochastic integrals.
			In such cases, the set $\cI$ is not a linear
			space.}}
	Other typical choices for $\cA_{adm}$
	would consist of simple integrands only; when $S$ is a 
	continuous process and $\cA_{adm}$ is the set of 
	process with finite variation then the above integral 
	can be defined through integration by parts (see \cite{DoS13,dolinsky2015martingale}).
\end{exm} 
In general, the absence of a common prior  poses some non-trivial technical questions about the integrability of contingent claims and net trades.
Clearly, it is possible to restrict the commodity space to  the class of bounded measurable functions (that are integrable with respect to any prior).
The condition  $\cI\subset\cH$  could be restrictive in some applications and we provide a way to overcome this difficulty in Appendix \ref{s.eh}.

\paragraph{Viability}
\label{s.viability}

In \cite{HarrisonKreps79, Kreps81}, viability is defined with the help of a single representative agent with a \emph{strictly} monotone preference relation. We allow for many agents, with \emph{weakly} monotone preferences, yet the family as a whole satisfies a strict monotonicity condition (Eqn.( \ref{EqnViable2})).
We show that one can prove the same results as Harrison and Kreps in classic cases, and one obtains easily the equivalence result in new cases involving Knightian uncertainty as in the Epstein-Ji model (Example \ref{ex.G}).
From the economic point of view, our definition of viability does not contradict the intuitive idea of an economic equilibrium because that idea does allow for  many agents and weakly monotone preferences, of course.
We thus believe that   the use of our new definition is justified (and its usefulness is illustrated by the proofs of our abstract theorems).

The reader might note that   our notion of equilibrium does not model   endowments explicitly as we assume   that the zero trade is optimal for each agent. This reduced approach comes  without loss of generality in our context.  In general, an agent is given by a preference relation $\succeq \in \cA$ and an endowment $e \in \cH$. Given the set of net trades, the agent chooses 
$\ell^* \in \cI$ such that
$ e + \ell^* \succeq e + \ell$ for all $\ell \in \cI$. By suitably modifying the preference relation, this can be reduced to the optimality of the zero trade at the zero endowment for a suitably modified preference relation. Let
$ X \succeq' Y$ if and only if $ X+e+\ell^* \succeq Y+e+\ell^*$.  It is easy to check that $\succeq'$ is also an admissible preference relation. 
For the new preference relation $\succeq'$, we then have
$0 \succeq' \ell$ if and only if $e+\ell^* \succeq e+\ell^*+\ell$. As $\cI$ is a cone, 
$\ell+\ell^*\in \cI$, and we conclude that we have indeed   $0 \succeq' \ell$
for all $\ell \in \cI$.

\paragraph{Sublinear Expectations}

Our fundamental theorem of asset pricing characterizes the absence of arbitrage with the help of a non--additive expectation $\cE$.
In decision theory, non--additive probabilities have a long history; \cite{Schmeidler89} introduces an extension of expected utility theory based on non--additive probabilities. The widely used max-min expected utility model
of \cite{GilboaSchmeidler89} is another instance.  If we define the subjective expectation of a payoff to be the minimal expected payoff over a class of priors, then the resulting notion of expectation has some of the common properties of an expectation like monotonicity and preservation of constants, but is no longer  additive. 

In our case, the non--additive expectation has a more objective than subjective flavor because it describes the pricing functional of the market. 
It assigns a nonpositive value to all  net trades; in this sense,  net trades have  the (super)martingale property under this expectation.   If we assume for the sake of the discussion that the set of net trades is a linear subspace, then the pricing functional has to be additive over that subspace. As a consequence, the value of all net trades under the sublinear pricing expectation is zero. For contingent claims that lie outside the marketed subspace, the pricing operation of the market is sub--additive. 

Whereas an additive probability measure is sufficient to characterize viable markets in probabilistic models, under uncertainty, it is more sensible to consider non--additive notion of expectation\footnote{\cite{BeissnerRiedel19} develop a general equilibrium model based on such non--additive pricing functionals.} as, for example, in the framework of Example \ref{ex.Ellsberg} because it allows to characterize fully  the ambiguity of   market prices\footnote{In some technically complex models, only \emph{sub}linear functionals can be strictly positive.}.

\paragraph{Relevant Claims}

We use the notion of relevant claims to generalize the typical approach to define arbitrage as positive net trades.  This approach introduces some additional flexibility and allows to cover    variants of the notion of arbitrage that were discussed in the literature.
For example, if some positive claim cannot be liquidated without costs, agents would  not consider a net trade that achieves such a payoff  as free lunch if the liquidation costs are larger than the potential gains. 
It is then reasonable to consider as relevant only a restricted class of positive claims, possibly  only cash.
Moreover,  relevant payoffs   identify  those  nonnegative consumption plans that some market participants strictly prefer to the null plan. 
The commodity spaces that are used to model markets with Knightian uncertainty
may be quite large, e.g. when we work with the space of all bounded, measurable functions.  In such models, it makes sense to work with a 
set of relevant claims that is smaller than the cone of positive claims, see also Example \ref{ex.G}.


We illustrate the usefulness of the notion of relevant claims  by relating our work to recent results in Mathematical Finance. Our approach gives a microeconomic foundation to the characterization of absence of arbitrage in ``robust'' or ``model--free'' finance.

Let  $\Om$ be a 
metric space with the   pointwise order $\oleq$.
In the finance literature, this approach is called
{\emph{model-independent}} as it does not rely on any probability measure. There is still a model, of course, given by $\Omega$ and the pointwise order.

A claim is nonnegative, $X \in \cP$, if $X(\om) \ge 0$
for every $\om \in \Om$
and $R \in \cP^+$ if $R \in \cP$ and there
exists $\om_0 \in \Om$ such that $R(\om_0)>0$.

In the literature
several different notions of arbitrage have been used.
Our framework allows to unify these different approaches under one framework with the help of the notion of relevant claims\footnote{One might also compare  the similar approach in
	\cite{BFM}.}.
We start with the following large
set of relevant claims
$$
\cR_{op}:= \cP^+= \left\{ R  \in \cP\ :\ \exists \om_0\in \Om\ 
{\mbox{such that}}\ 
R(\om_0)> 0  \ \right\}.
$$ 
With this notion of relevance, an investment opportunity $\ell$ is 
an arbitrage if $\ell(\om) \ge0$ for every $\om$ with a strict inequality for some $\omega$, corresponding to  the notion of
{\emph{one point arbitrage}} considered in   \cite{Riedel14}.
In this setting, no arbitrage is equivalent to the existence 
a set of martingale measures $\cQ_{op}$
so that for each point there exists $\Q \in \cQ_{op}$ putting
positive mass to that point.

Other authors (\cite{BFM,Riedel14,dolinsky2014robust}) introduced the notion of \emph{open arbitrage}. 
In probability space frameworks, the family of null-sets defines ``small``, i.e. negligible events.  In absence of a reference probability, it might be still reasonable to distinguish large and small events. When there is a topology, one can define  ``small" events as those sets that are countable unions of closed sets with empty interior (Baire first category set). Open sets are then considered as relevant.  We might then call  $\ell \in \cI$  an {\emph{open arbitrage}} if it is nonnegative and is strictly positive 
on an open set. This case can be modeled in our framework by requiring
the relevant claims to be continuous, nonnegative  functions that are different from zero somewhere, i.e.,
$$
\cR_{open}:= \left\{ R  \in C_b(\Om) \cap \cP\ :\ \exists \om_0\in \Om\ 
{\mbox{such that}}\ 
R(\om_0)> 0  \ \right\}.
$$
It is then clear that when $R \in \cR$
then it is non-zero on an open set.

\cite{Acciaioetal} 
define a claim to be an arbitrage when it is positive everywhere, corresponding, in our model, to the choice $$
\cR_{+}:= \left\{ R  \in \cP \ :\ \
R(\om)> 0,  \ \forall \ \om \in \Om\ \right\}.
$$
\cite{C} consider a slightly stronger
notion of relevant claims.  Their choice is 
\begin{equation}
	\label{e.ru}
	\cR_{u}=\left\{ R\in \cP \ :\ \exists c \in (0, \infty)\
	{\mbox{such that}}\ R \equiv c\ \right\}.
\end{equation}
Hence, $\ell \in \cI$ is an arbitrage if is 
uniformly positive, which is sometimes 
called {\emph{uniform arbitrage}}.
Notice that with the choice $\cR_u$, 
the notions of arbitrage and free lunch with vanishing
risk are equivalent\footnote{
	The no arbitrage condition with $\cR_u$ is the weakest one, 
	$\cR_{op}$ is the strongest one.  The first one
	is equivalent to the existence of one
	sublinear martingale expectation.  The latter one   is equivalent to the existence
	of a sublinear expectation that puts positive measure 
	to all points.
	In general, the no-arbitrage condition based on  $\cR_+$ is
	not equivalent to the absence of  uniform arbitrage.
	However,  absence of uniform arbitrage implies the 
	existence of a linear bounded functional that is 
	consistent with the market.  
	In particular, 
	risk neutral functionals  are positive on $\cR_u$. 
	Moreover, if the set $\cI$ is ``large'' enough then
	one can show that the risk neutral functionals give rise to  countably additive measures.  In \cite{Acciaioetal}, this conclusion  is achieved by using
	the so-called ``power-option'' placed in the set $\cI$
	as a static hedging possibility, compare also \cite{C}.}.

\section{Proof of the Theorems}
\label{s.proofsmain}
Let $(\cH,\tau, \le\cI,\cR)$ be a given financial market. 
Recall that $(\cH,\tau)$ is a metrizable topological vector space; we write $\cH^\prime$ for its topological dual. We let $\cH^\prime_+$ be the set of all positive
functionals, i.e., $\varphi \in \cH^\prime_+$
provided that $\varphi(X) \ge 0$ for every 
$X \geq 0$ and $X \in \cH$. 

The following functional generalizes the notion of super-replication functional from the probabilistic to our order-theoretic framework. It  plays a central role
in our analysis. For $X \in \cH$, let
\begin{align}
	\label{e.v}
	\cD(X):=\inf \{ \  c\in \mathbb{R}\ : \ 
	\exists  \{\ell_n\}_{n=1}^{\infty} \subset \cI, \ 
	& \{e_n\}_{n=1}^{\infty} \subset \cH_+,\  e_{n}\stackrel{\tau}{\rightarrow}0,
	\\\nonumber
	&\  {\text{such that}}\ \
	\left.  c+e_n +\ell_n \ge X \right\}.
\end{align}
Note that $\cD$ is extended real valued.  In particular,
it takes the value $+\infty$ when there are no super-replicating
portfolios. It might also take the value $-\infty$ if
there is no lower bound. 

We observe first that the absence of free lunches with vanishing 
risk  can be equivalently described by the statement that the super-replication functional $\cD$ assigns a strictly positive value to all relevant claims.

\begin{prp}
	\label{p.viable}
	The financial market is strongly free of arbitrage if and only if
	$\cD(R) >0$ for every $R \in \cR$.
\end{prp}
\begin{proof}
	Suppose $\{\ell_n\}_{n=1}^\infty \subset \cI$ is a free lunch with vanishing risk.
	Then, there is $R^* \in \cR$ and $\{e_n\}_{n=1}^{\infty} \subset \cH_+$ with $e_{n}\stackrel{\tau}{\rightarrow}0$ so that $e_n+\ell_n\ogeq R^*$.
	In view of the definition, we obtain
	$\cD(R^*) \le 0$.
	
	To prove the converse, suppose that $\cD(R^*) \le 0$ for some
	$R^* \in \cR$.  Then, the definition of $\cD(R^*)$ implies that there is a sequence of real numbers $\{c_k\}_{k=1}^{\infty}$ with $c_k\downarrow\cD(R^*)$, net trades  $\{\ell_{k,n}\}_{n=1}^{\infty}\subset\cI$,  and   $\{e_{k,n}\}_{n=1}^{\infty}\subset\cH_+$ with $e_{k,n}\stackrel{\tau}{\rightarrow}0$ for $n\rightarrow\infty$ such that
	$$
	c_k+e_{k,n}+\ell_{k,n}\ge R^*,\qquad \forall\ n,k\in\N.
	$$
	Let $B_{r}(0)$ be the ball with radius $r$ centered at zero with the metric compatible with $\tau$.
	For every $k$, choose  $n=n(k)$ such that $e_{k,n}\in B_{\frac{1}{k}}(0)$.
	Set $\tilde{\ell}_k:=\ell_{k,n(k)}$ and $\tilde{e}_k:=e_{k,n(k)}+(c_k\vee 0)$.
	Then,
	$\tilde{e}_k+\tilde{\ell}_k\ogeq R^*$ for every $k$.  
	Since  $\tilde{e}_k\stackrel{\tau}{\rightarrow}0$, $\{\tilde{\ell}_k\}_{k=1}^\infty$ is a free lunch
	with vanishing risk.
\end{proof}

It is clear that $\cD$ is convex and
we  now use the  tools of convex duality to characterize this functional in more detail.   
Recall the set of absolutely continuous martingale functionals $\cQ_{ac}$
defined in Section \ref{s.setup}.

\begin{prop}
	\label{p.convexdual} 
	Assume that the financial market is strongly free of arbitrage. Then, the super-replication functional $\cD$ defined in \eqref{e.v} is a 
	lower semi-continuous, sublinear martingale expectation with full support.
	Moreover,
	\begin{equation*}
		\cD(X)= \sup_{\varphi \in  \cQ_{ac}}\ 
		\varphi(X), \quad X \in \cH.
	\end{equation*}
\end{prop}

The technical proof of this statement can be found in Appendix \ref{ss.super-replication}. The important insight is that the super-replication functional can be described by a family of linear functionals. In the probabilistic setup, they correspond to the family of (absolutely continuous) martingale measures.
With the help of this duality, we are now able to carry out the proof of our first main theorem.

\begin{proof}[Proof of Theorem \ref{t.viable}]
	Suppose first that the market is viable
	and for some  $R^*\in \cR$,  there are sequences $\{e_n\}_{n=1}^{\infty} \subset \cH_+$ and $\{\ell_n\}_{n=1}^{\infty} \subset \cI$ with $e_{n}\stackrel{\tau}{\rightarrow} 0$, 
	and  $e_n+\ell_n \ge R^*$.
	By viability, there is   a family of agents  $\{\preceq_a\}_{a\in A}\subset \cA$ 
	such that for some $a \in A$ we have $R^*\succ_a 0$.
	Since $\le$ is a pre-order compatible with the vector space operations,
	we  have
	$-e_n  + R^* \le \ell_n$.
	As   $\preceq_a \,\in \cA$ is monotone
	with respect to $\le$, we have
	$-e_n  + R^* \preceq_a \ell_n$.
	By optimality of the zero trade,  $\ell_n \preceq_a 0$, and we get $-e_n  + R^* \preceq_a 0$. 
	By lower semi--continuity of $\preceq_{a}$, we conclude that
	$R^* \preceq_a 0$,
	a contradiction.
	
	Suppose now that the market is strongly free of arbitrage. 
	By Proposition \ref{p.viable},
	$\cD(R) > 0$,
	for every $R\in \cR$.  
	In particular, this implies that
	the family  $\cQ_{ac}$ is non-empty,
	as otherwise the supremum over $\cQ_{ac}$ would be $-\infty$.
	For each $\varphi \in \cQ_{ac}$,
	define $\preceq_\varphi$ by,
	$$
	X \preceq_\varphi Y , \quad
	\Leftrightarrow
	\quad
	\varphi(X) \le \varphi(Y).
	$$
	One directly verifies that $\preceq_\varphi \in \cA$.
	Moreover, $\varphi(\ell) \le \varphi(0)=0$ for any $\ell\in\cI$ implies that $\ell^*_\varphi=0$ is optimal for $\preceq_{\varphi}$ and \eqref{EqnViable1} is satisfied.
	Finally, Proposition \ref{p.viable} and Proposition \ref{p.convexdual} imply that for any $R\in\cR$, there exists 
	$\varphi\in \cQ_{ac}$ such that $\varphi(R)>0$; thus, \eqref{EqnViable2} is satisfied.
	We deduce that   $\{\preceq_\varphi\}_{\varphi \in \cQ_{ac}}$  supports the financial market  $(\cH,\tau, \oleq,\cI,\cR)$.
\end{proof}

The previous 
arguments also imply our version of the fundamental theorem of asset pricing. In fact, with absence of arbitrage, the super-replication function is a lower semi-continuous sublinear martingale expectation with full support. Convex duality allows to prove the converse.

\begin{proof}[Proof of Theorem \ref{t.subexp}]
	Suppose the market is viable. From Theorem \ref{t.viable}, it is strongly free of arbitrage. 
	From Proposition \ref{p.convexdual}, the super-replication functional  is the desired lower semi-continuous sublinear martingale expectation with full support.

	Suppose now that $\cE$ is a lower semi-continuous sublinear martingale expectation with full support. 
	In particular, $\cE$ is a convex, lower semi-continuous, proper functional.  Then,  by
	the Fenchel-Moreau theorem,
	$$
	\cE(X)= \sup_{\varphi \in dom(\cE^*)} \varphi(X),
	$$
	where
	$dom(\cE^*)=\left\{
	\varphi \in \cH^\prime\ :\
	\varphi(X) \le \cE(X), \ \forall \ X \in \cH
	\ \right\}$. 
	We can proceed as in the proof of Theorem \ref{t.viable},
	to verify the viability of $\mkt$ using
	the preference relations $\{\preceq_\varphi\}_{\varphi\in dom(\cE^*)}$.
	
	We finally show the maximality of $\cE_{\cQ_{ac}}$. Let $\cE$ be a lower semi-continuous sublinear martingale expectation with full support. With the help of the martingale property of $\cE$ one can show, as in Lemma \ref{l.emm}, that every $\varphi \in dom(\cE^*)$ is a  martingale functional. As $\cE$ is monotone with respect to $\le$, we also conclude that $\varphi$ vanishes for negligible payoffs. Hence, we obtain $dom(\cE^*) \subset\cQ_{ac}$. From the above dual representation for $\cE_{\cQ_{ac}}$,   $\cE(X)\le \cE_{\cQ_{ac}}(X)$ for every $X\in\cH$ follows.
\end{proof}

\section{Conclusion}

This paper studies      economic viability of a given financial market without assuming   a probability-space framework.
We show that it is possible to  understand the equivalence of economic viability and the absence of arbitrage on the basis  solely of a common  order; the order (which is typically quite incomplete)   is unanimous in the sense that agents' preferences are monotone with respect to it. 
A given financial market is viable if and only if a \emph{sublinear} pricing functional exists that is consistent with the given asset prices. 

The properties of the common order are reflected in expected equilibrium returns.  When the common order is given  by the expected value under some common prior,  expected returns under that prior have to be equal in equilibrium, and thus,  Fama's Efficient Market Hypothesis results. 
If  the common order is determined by  the almost sure order  under some common prior, we obtain the 
weak form of the efficient market hypothesis that states that  expected returns are equal under some  
(martingale) measure that shares the same null sets as the common prior. 

In situations of Knightian uncertainty, it might be too demanding to impose a common prior for all agents. When Knightian uncertainty is described by a class of priors,  it is necessary  to replace the linear (martingale) expectation by a  sublinear expectation. It is then no longer possible to reach the conclusion that expected returns are equal under some probability measure. Knightian uncertainty might thus be an explanation for empirical violations of the Efficient Market Hypothesis.  In particular, there is always a range of economically justifiable arbitrage--free prices. 
In this sense,  Knightian uncertainty shares  similarities with markets with friction
or that are incomplete, but the economic reason for the 
price indeterminacy is different. 


\begin{appendix}
	
	\label{s.appendix}

	\section{Proof of Proposition \lowercase{\ref{p.convexdual}}}
	\label{ss.super-replication}
	
	We separate the proof in several steps.
	Recall that the super-replication functional $\cD$ is defined in \eqref{e.v}.
	
	\begin{lem}
		\label{l.general} Assume that the financial market is strongly free of arbitrage. Then,
		$\cD$  is convex, lower semi--continuous and 
		$\cD(X)>-\infty$ for every $X \in \cH$.
	\end{lem}
	\begin{proof}
		The convexity of $\cD$ follows immediately from the definitions.
		To prove lower semi-continuity, consider a sequence
		$X_k\stackrel{\tau}{\rightarrow}X$ with $\cD(X_k)\le c$.
		Then, by definition, for every $k$ there exists a sequence  
		$\{e_{k,n}\}_{n=1}^{\infty}\subset\cH_+$ with $e_{k,n}\stackrel{\tau}{\rightarrow} 0$ 
		for $n\rightarrow\infty$ and a sequence $\{\ell_{k,n}\}_{n=1}^{\infty}\subset\cI$ 
		such that
		$c+\frac{1}{k}+e_{k,n}+\ell_{k,n}\ge X_k$, for every $k,n$.
		Let $B_{r}(0)$ be the ball of radius $r$ centered around zero
		in the metric compatible with $\tau$.  Choose $n=n(k)$ such that 
		$e_{k,n}\in B_{\frac{1}{k}}$
		and set  $\tilde{e}_k:=e_{k,n(k)}$, $\tilde{\ell}_k:=\ell_{k,n(k)}$.
		Then, $ c+\frac{1}{k}+\tilde{e}_k+(X-X_k)+\tilde{\ell}_k\ge X$
		and $\frac{1}{k}+\tilde{e}_k+(X-X_k)\stackrel{\tau}{\rightarrow}0$ as $k\rightarrow\infty$.
		Hence, $\cD(X)\le c$.  This proves that $\cD$ is lower semi-continuous.
		
		The constant claim $1$ is relevant and by Proposition \ref{p.viable},
		$\cD(1) \in (0,1]$; in particular, it is finite.
		Towards a counter-position, suppose that there exists $X\in\cH$ such that $\cD(X)=-\infty$. 
		For $\lambda \in [0,1]$, set $X_\lambda:=X+\lambda(1-X)$.
		The convexity of $\cD$ implies that $\cD(X_\lambda)=-\infty$ for every $\lambda\in [0,1)$. Since $\cD$ is lower semi-continuous,  
		$0<\cD(1)\le\lim_{\lambda\to 1} \cD(X_\lambda)=-\infty$,
		a contradiction.
	\end{proof}

	\begin{lem}
		\label{l.convex} Assume that the financial market is strongly free of arbitrage. The super-replication functional $\cD$
		is a sublinear expectation with full-support. Moreover,
		$\cD(c)=c$ for every $c\in\R$, and
		\begin{equation}
			\label{e.ell}
			\cD(X+ \ell) \le \cD(X), \quad \forall \ \ell \in \cI, \ X \in \cH.
		\end{equation}
		In particular, $\cD$ has the martingale property.
	\end{lem}
	\begin{proof} 
		We prove this result in two steps.
		
		{\emph{Step 1.}} In this step we prove that $\cD$
		is a sublinear expectation.
		Let $X,Y \in \cH$ such that $X\oleq Y$. Suppose that there are $c \in \R$, $\{\ell_n\}_{n=1}^\infty \subset \cI$ and $\{e_n\}_{n=1}^{\infty} \subset \cH_+$ with $e_{n}\stackrel{\tau}{\rightarrow}0$ satisfying, $Y\oleq c + e_n+ \ell_n$.  Then, from the transitivity of 
		$\oleq$, we also have $X\oleq c + e_n+ \ell_n$.
		Hence, $\cD(X) \le \cD(Y)$ and consequently
		$\cD$ is monotone with respect to $\oleq$.
		
		Translation-invariance,
		$\cD(c+ g)= c + \cD(g)$, follows directly
		from the definitions.
		
		We next show that $\cD$
		is sub-additive. Fix $X,Y \in \cH$.
		Suppose that either $\cD(X)=\infty$ or 
		$\cD(Y)=\infty$.  Then, 
		since by Lemma \ref{l.general}
		$\cD>-\infty$, we have $\cD(X)+ \cD(Y)  = \infty$
		and the sub-additivity follows directly.
		Now we consider the case
		$\cD(X), \cD(Y) < \infty$. 
		Hence,
		there are  $c_X, c_Y \in \R$, $\{\ell_n^X\}_{n=1}^{\infty},\{\ell_n^Y\}_{n=1}^\infty \subset \cI$ and 
		$\{e^X_n\}_{n=1}^{\infty},\{e^Y_n\}_{n=1}^{\infty} \subset \cH_+$ with 
		$e^X_n,e^Y_{n}\stackrel{\tau}{\rightarrow}0$ satisfying,
		$$
		c_X + \ell_n^X +e^X_n\ogeq  X ,\quad
		c_Y + \ell_n^Y+e^Y_n \ogeq Y .
		$$
		Set  $\bar c: = c_X+c_Y$,
		$\bar \ell_n:= \ell_n^X+\ell_n^Y$, $\bar e_n:= e^X_n+e^Y_n$ .
		Since $\cI$ is a positive
		cone, $\{\bar \ell_n\}_{n=1}^{\infty} \subset \cI$, $\bar e_{n}\stackrel{\tau}{\rightarrow}0$ and
		$$
		\bar c +\bar{e}_n +\bar \ell_n  \ogeq   X +Y\quad
		\Rightarrow
		\quad
		\cD(X+Y) \le \bar c.
		$$
		Since this holds for any such $c_X, c_Y$,
		we conclude that
		$$
		\cD(X+Y) \le  \cD(X) + \cD(Y).
		$$
		
		Finally we show that $\cD$ is positively homogeneous
		of degree one. Suppose that $c+e_n+\ell_n \ogeq  X$ for some
		constant $c$, $\{\ell_n\}_{n=1}^\infty \subset \cI$ and $\{e_n\}_{n=1}^{\infty} \subset \cH_+$ with $e_{n}\stackrel{\tau}{\rightarrow}0$.
		Then, for any $\lambda >0$ and for any $n\in\N$, 
		$\lambda c + \lambda e_n+ \lambda \ell_n \ogeq \lambda X$.
		Since $\lambda \ell_n \in \cI$ and $\lambda e_{n}\stackrel{\tau}{\rightarrow}0$, this implies that
		\begin{equation}
			\label{e.ph}
			\cD(\lambda X ) \le \lambda \ \cD(X), \quad\ \lambda >0, \ X \in \cH.
		\end{equation}
		Notice that above holds trivially
		when $\cD(X)= +\infty$.
		Conversely, if $\cD(\lambda X )=+\infty$ we are done. Otherwise,
		we use \reff{e.ph} with $\lambda X$ and $1/\lambda$,
		$$
		\cD(X)= \cD\left(\frac1{\lambda} \lambda X\right) \le
		\frac1{\lambda} \cD(\lambda X), \quad
		\Rightarrow
		\quad
		\lambda \cD(X) \le \cD(\lambda X).
		$$
		Hence, $\cD$ positively homogeneous
		and it is a sublinear expectation.\\
		
		{\emph{Step 2.}} In this step, 
		we assume that the financial market is
		strongly free of arbitrages. 
		Since $0\in\cI$, we have $\cD(0)\leq 0$. 
		If the inequality is strict we obviously have a free lunch with vanishing risk, 
		hence $\cD(0)= 0$ and from translation-invariance the same 
		applies to every $c\in\R$. Moreover, by Proposition \ref{p.viable}, 
		$\cD$ has full support. 
		Thus, we only need to prove \reff{e.ell}.
		
		Suppose that $X \in \cH, \ell \in \cI$ and $c+e_n+\ell_n^X \ogeq X$. Since
		$\cI$ is a convex cone, $\ell_n^X+\ell \in \cI$ 
		and $c+e_n+(\ell + \ell_n^X)\ogeq X+\ell$.
		Therefore,
		$\cD(X+\ell) \le c$.  Since this holds for all
		such constants, we conclude that $\cD(X+\ell) \le \cD(X)$
		for  all $X \in \cH$. In particular  $\cD(\ell) \le 0$ and the martingale property is satisfied.
	\end{proof}

	\begin{rem}
		\label{rem:classicSH}
		Note that for $\cH=(\bdd, \|\cdot\|_\infty)$, the definition of $\cD$ reduces to the classical one:
		\begin{equation}\label{e.v.classical}
			\cD(X):=\inf \left\{\ c\in \R\ :\
			\exists \ \ell \in \cI,\ {\mbox{such that}}\
			c+\ell \ogeq X\ \right\}.
		\end{equation}
		
		Indeed, if $c+\ell \ogeq X$ for some $c$ and $\ell$, one can use the constant sequences $\ell_n \equiv\ell$ and $e_n\equiv 0$ to get that $\cD$ in \eqref{e.v} is less or equal than the one in \eqref{e.v.classical}. 
		For the converse inequality observe that if  $c+e_n +\ell_n \ge X$ for some $c,\ell_n$ and $e_n$ with $\|e_n\|_\infty\rightarrow 0$, then the infimum in \eqref{e.v.classical} is less or equal than $c$. The thesis follows. Lemma \ref{l.general} is in line with the well known fact that the classical super-replication functional in $\bdd$ is Lipschitz continuous with respect to the sup-norm topology.
	\end{rem}
	
	The results of Lemma \ref{l.general} and Lemma \ref{l.convex} imply that
	the super-replication functional defined in \eqref{e.v} is a proper convex function
	in the language of convex analysis, compare, e.g., \cite{Rockafellar}. By the 
	classical Fenchel-Moreau theorem, we have the
	following dual representation of $\cD$,
	\begin{align*}
		\cD(X)&= \sup_{\varphi \in \cH^\prime}\ 
		\left\{ \varphi(X) - \cD^*(\varphi)\right\}, \quad X \in \cH,\quad{\mbox{where}}\\
		\cD^*(\varphi)&= \sup_{Y \in \cH}\ 
		\left\{ \varphi(Y) - \cD(Y)\right\}, \quad \varphi \in \cH^\prime.
	\end{align*}
	Since $\varphi(0)=\cD(0)=0$, 
	$\cD^*(\varphi) \ge \varphi(0)
	-\cD(0)=0$ for every $\varphi \in \cH^\prime$.  
	However, it may take the
	value plus infinity. Set, 
	$$
	dom(\cD^*):= \left\{\ \varphi \in \cH^\prime\ : 
	\ \cD^*(\varphi)< \infty\right\}.
	$$
	
	\begin{lem}
		\label{l.pzero} We have
		\begin{equation}
			\label{e.domain}
			dom(\cD^*)= \left\{\ \varphi \in  \cH^\prime_+\ : 
			\ \cD^*(\varphi)=0\right\}
			= \left\{\ \varphi \in  \cH^\prime_+\ : 
			\ \varphi(X) \le \cD(X ), \ \ \forall \ X \in \cH \right\}  .
		\end{equation}
		In particular,
		\begin{equation*}
			\cD(X)= \sup_{\varphi \in  dom(\cD^*)}\ 
			\varphi(X), \quad X \in \cH.
		\end{equation*}
		Furthermore, there are free lunches with vanishing risk 
		in the financial market, whenever $ dom(\cD^*)$ is empty.
	\end{lem}
	\begin{proof} Clearly the two sets on the right of \reff{e.domain}
		are equal and included in $dom(\cD^*)$.
		The definition of $\cD^*$ implies that
		$$
		\varphi(X) \le \cD(X ) + \cD^*(\varphi), \quad
		\forall \ X \in \cH, \ \varphi \in \cH^\prime.
		$$
		By  homogeneity,
		$$
		\varphi(\lambda X) \le \cD(\lambda  X) + \cD^*(\varphi),\quad
		\Rightarrow
		\quad
		\varphi(X) \le \cD(X) + \frac1{\lambda} \cD^*(\varphi),
		$$
		for every $\lambda >0$ and $X \in \cH$.  
		Suppose that $\varphi \in dom(\cD^*)$.
		We then let $\lambda$ go to infinity to arrive at
		$\varphi(X) \le \cD(X)$ for all $X \in \cB_b$.  Hence,
		$\cD^*(\varphi)=0$.   
		
		Fix $X \in \cH_+$.
		Since $\oleq$ is monotone
		with respect to the pointwise order, $-X \oleq 0$.
		Then, by the monotonicity of $\cD$,
		$\varphi(-X) \le \cD(-X)\le \cD(0)\leq 0$.  Hence,
		$\varphi \in \cH^\prime_+$.
		
		\vspace{3pt}
		
		Now suppose that $dom(\cD^*)$ is empty or, equivalently, 
		$\cD^*\equiv \infty$.  Then, the dual
		representation implies that $\cD \equiv -\infty$. 
		In view of Proposition \ref{p.viable},
		there are free lunches with vanishing risk
		in the financial market.
	\end{proof}
	
	We next show that, under the assumption of absence of free 
	lunch with vanishing risk with respect to 
	any  $\cR$, the set $ dom(\cD^*)$ is equal to
	$\cQ_{ac}$ defined in Section \ref{s.setup}.
	Since any relevant set $\cR$ by hypothesis
	contains $\cR_u$ defined in \reff{e.ru},
	to obtain this conclusion it would be 
	sufficient to assume the
	absence of free 
	lunch with vanishing risk with respect to 
	any  $\cR_u$.

	\begin{lem}
		\label{l.emm}
		Suppose the financial market is strongly free of arbitrage with respect to $\cR$.
		Then, $dom(\cD^*)$ is equal to the set of absolutely continuous martingale 
		functionals $\cQ_{ac}$.
	\end{lem}
	\begin{proof} The fact that  $dom(\cD^*)$ is non-empty follows from Lemma \ref{l.convex} and Lemma \ref{l.pzero}. Fix an arbitrary $\varphi \in dom(\cD^*)$.
		By Lemma \ref{l.convex}, $\cD(c)=c$ for every constant $c\in\R$.
		In view of the dual representation of Lemma \ref{l.pzero},
		$$
		c \varphi(1)= \varphi(c) \le \cD(c)=c, 
		\quad \forall c \in \R.
		$$
		Hence, $\varphi(1)=1$.  
		
		We continue by proving the monotonicity property.
		Suppose that 
		$X\in \cP$. Since $0\in\cI$, we obviously have $\cD(-X)\leq 0$. 
		The dual representation implies that $\varphi(-X) \le \cD(-X)\leq 0$. 
		Thus, $\varphi(X) \ge 0$.
		
		We now prove the supermartingale property. Let $\ell \in \cI$. 
		Obviously $\cD(\ell)\leq 0$. By the dual representation,
		$\varphi(\ell ) \le \cD(\ell) \le 0$.
		Hence $\varphi$ is a martingale functional.
		The absolute continuity follows
		as in Lemma \ref{l.ez}.
		Hence, $\varphi \in \cQ_{ac}$.

		To prove the converse, fix an arbitrary  $\varphi \in \cQ_{ac}$.  
		Suppose that $X \in \cH$, $c \in \R$, $\{\ell_n\}_{n=1}^\infty \subset \cI$ and $\{e_n\}_{n=1}^{\infty} \subset \cH_+$ with $e_{n}\stackrel{\tau}{\rightarrow}0$ satisfy,
		$c + e_n + \ell_n \ogeq  X$.  From the properties
		of $\varphi$,
		$$
		0 \le \varphi(c + e_n + \ell_n - X) = \varphi(c +e_n -X) + \varphi(\ell_n)
		\le c - \varphi(X-e_n).
		$$
		Since $e_{n}\stackrel{\tau}{\rightarrow}0$ and $\varphi$ is continuous, $\varphi(X) \le \cD(X)$ for every $X \in \cH$.
		Therefore, $\varphi \in dom(\cD^*)$.
	\end{proof}
	
	\begin{proof}[Proof of Proposition \ref{p.convexdual}]
		It follows directly from Lemma \ref{l.pzero} and Lemma \ref{l.emm}.
	\end{proof}
	
	\vspace{5pt}
	We have the following immediate corollary, which  
	proves the first part of the Fundamental Theorem of Asset Pricing in this context.
	\begin{cor}
		\label{c.ftap}
		The financial market
		is strongly free
		of arbitrage if and only $\cQ_{ac}\neq\emptyset$ and for any $R\in\cR$, there exists $\varphi_R\in\cQ_{ac}$ such that $\varphi_R(R)>0$.
		
	\end{cor}
	\begin{proof} 
		By contradiction, suppose that there exists $R^*$ such that $e_n + \ell_n \ogeq  R^*$ with $e_n\stackrel{\tau}{\rightarrow}0$. Take $\varphi_{R^*}$ such that $\varphi_{R^*}(R^*)>0$ and observe that
		$0<\varphi_{R^*}(R^*) \le \varphi(e_n + \ell_n) \le \varphi(e_n)$. Since $\varphi\in\cH'_+$, $\varphi(e_n)\to 0$ as $n\to\infty$, which is a contradiction.
		
		In the other direction, assume that the financial market is strongly free of arbitrage.
		By Lemma \ref{l.emm}, $dom(\cD^*)=\cQ_{ac}$. Let $R \in \cR$ and note that, by Proposition \ref{p.viable}, $\cD(R)>0$.
		It follows that there exists
		$\varphi_R \in dom(\cD^*)=\cQ_{ac}$ satisfying  $\varphi_R(R)>0$. 
	\end{proof}
	
	\begin{rem}
		The set of positive functionals
		$\cQ_{ac} \subset \cH^\prime_+$ is the analogue of
		the set of local martingale measures
		of the classical setting.  
		Indeed, all elements of $\varphi \in \cQ_{ac}$
		can be regarded as  supermartingale ``measures'',
		since $\varphi(\ell) \le 0$ for every $\ell \in \cI$.
		Moreover, the property $\varphi(Z)=0$ for every $Z \in \cZ$
		can be regarded as absolute 
		continuity with respect to null sets.  
		The full support property is our analog to the  converse absolute 
		continuity. 
		However, the full-support property cannot be achieved by a single
		element of $\cQ_{ac}$.

		\cite{BN13} study arbitrage for  a set of priors $\cM$. 
		The absolute continuity and the full support properties then translate to
		the statement that ``$\cM$ and $\cQ$ have the same polar sets''. 
		In the paper by  \cite{BFM},
		a class of relevant sets $\cS$ is given and the two properties 
		can summarised by the statement
		``the set $\cS$ is  not contained in the polar sets of $\cQ$''.
		
		Also, when $\cH=\bdd$, $\cH^\prime$ is the class of bounded additive measures $ba$.  It is a classical
		question whether one can restrict
		$\cQ$ to the set of countable additive
		measures $ca_r(\Om)$. 
		In several of the examples
		described in Section \ref{s.discussion} and \ref{s.examples} this is proved.
		However, there are 
		examples for which this is not true.
	\end{rem}
\section{Extensions}
\label{s.eh}
For the notation and definitions, we refer to the main text. 
Let $\cB(\Om,\cF)$ be the set of all $\cF$ measurable real-valued functions on $\Om$.
Any Banach space 
contained in $\cB(\Om,\cF)$ satisfies the requirements for $\cH$. 
In our examples, we used the spaces 
$\cL^1(\Om,\cF,\P)$, $\cL^2(\Om,\cF,\P)$, $\cL^1(\Om,\cF,\cM)$ and   $\bdd(\Om,\cF)$,
the set of all bounded
functions in $\cB(\Om,\cF)$, with the supremum norm.
In the latter case, the super-hedging functional enjoys 
several properties  as discussed in Remark 2.3 in the main text .

Since we require that $\cI\subset\cH$, in the case of $\cH=\bdd(\Om,\cF)$
this means that all the trading instruments are bounded. This could be restrictive in some applications
and  we now provide another example that overcomes this difficulty.
To define this set,
fix $L^* \in \cB(\Om,\cF)$ with $L^*(\omega)\ge 1$ for every $\omega\in\Omega$.
Consider the linear space
$$
\bdl:=\left\{ X \in \cB(\Om,\cF)\ :\
\exists\ \alpha \in \R^+
{\text{ such that }}
|X(\omega)| \le \alpha L^*(\omega)\ \forall \omega\in\Omega \right\}
$$
equipped with the norm,
$$
\| X \|_\ell:= \inf\{ \alpha \in \R^+\ :\
|X(\omega)| \le \alpha L^*(\omega)\ \forall \omega\in\Omega\}= \left\| \frac{X}{L^*}\right\|_\infty.
$$ 
We denote the topology induced by this norm by $\tau_\ell$.
Then, $\bdl(\Om,\cF)$ with $\tau_\ell$
is a Banach space and satisfies our assumptions.
Note that if $L^*=1$, then $\bdl(\Om,\cF)=\bdd(\Om,\cF)$. 

Now, suppose that
\begin{equation}
	\label{e.ells}
	L^*(\om):= c^*+ \hat \ell(\om), \quad
	\ \om \in \Om,
\end{equation}
for some $c^*>0$, $\hat \ell \in \cI$.
Then,
one can define the
super-replication functional 
as in Remark A.3 of the main paper.

Another important extension is to relax the assumption that the consumption sets is equal to the entire space $\cH$.
This hypothesis is taken in the classical literature as well as in the main body of the paper but it might be restrictive in some applications. We show here that, within our framework, we may accommodate a smaller consumption set, in particular, we are able to restrict to consumption sets bounded from below. Let us fix the lower bound to be $0$ for the sake of discussion, which also correspond to the most relevant case of non-negative consumption. Consider
a market $(H, \tau, \oleq, \cI, \cR)$ where the topology is generated by open intervals with respect to the (strict) order; by definition, the set $O := \{X > 0\}$ is open. Given
a preference relation on $O$ we can extended it to the whole space by treating all elements of $\{X \oleq 0\}$ as indifferent and $X \prec Y$ if $X\notin O$ and $Y\in O$. Since the
preferences in our definition of $\cA$ are only required to be $\tau$-lower semi-continuity, this class satisfies all the required properties. In particular, the  class of linear agents constructed in the proofs of Theorem 2.1 and 2.2 in the main text may be modified accordingly: for a given linear continuous functional $\varphi$, we may set utility
to minus infinity on the complement of $O$. The induced preference relation is in $\cA$. Additionally, agents with power utilities
\[
U(X)=\E\left[\frac{(X+c)^{1-\gamma}}{1-\gamma}\right]
\] 
with a constant $c$ can be included in $\cA$. Restricting consumption to be positive may result in the failure of extendability of the pricing functional\footnote{We thank an anonymous referee for pointing out this aspect.}, therefore the classical theory of \cite{HarrisonKreps79,Kreps81} does not apply. On the contrary, since we do not insist on a single representative agent, this aspect does not affect the results of Section 2 of the main text.

\section{No Arbitrage versus  No Free-Lunch-with-Vanishing-Risk}
\label{s.versus}

Let $(\cH,\tau, \oleq,\cI,\cR)$ be a financial market.
An arbitrage opportunity is always a free lunch
with vanishing risk. The purpose of this section is to investigate when these two notions are equivalent.

\subsection{Attainment}
\label{ss.attainment}
\begin{dfn}
	\label{d.attainment}
	
	We say that a financial market has the \emph{attainment property},  if
	for every $X \in \cH$ there exists a minimizer in Equation (5.1) of the main text, i.e., there exists
	$\ell_X \in \cI$  satisfying,
	$$
	\cD(X) + \ell_X \ogeq X.
	$$
	
\end{dfn}

\begin{prp}
	\label{p.versus}  
	Suppose that a financial market has the attainment property. 
	Then,
	it is strongly free of arbitrage if and only if it has no arbitrages.
\end{prp}
\begin{proof} 
	Let $R^* \in \cR$.  By hypothesis, there exist
	$\ell \in \cI^*$ so that
	$\cD(R^*) + \ell^* \ogeq R^* $.
	If the market has no arbitrage, then we conclude that $\cD(R^*)>0$.
	In view of Proposition 5.1 of the main text, this proves 
	that the financial market is also strongly free of arbitrage.
	Since no arbitrage is weaker condition, 
	they are equivalent.
\end{proof}

\subsection{Discrete Time Markets with Finite Horizon}
\label{ss.fdc}

In this subsection and in the next section, we restrict 
ourselves to arbitrage considerations
in finite discrete-time markets.

We start by introducing a discrete filtration 
$\F:=(\cF_t)_{t=0}^T$ on subsets of $\Om$.
Let $S=(S_t)_{t=0}^T$ be an adapted
stochastic process{\footnote{ When working
		with $N$ stocks, a canonical choice
		for $\Om$ would be 
		$$
		\Om=\{ \om=(\om_0,\ldots,\om_T)
		\ :\ \om_i\in [0,\infty)^N, \ i=0,\ldots,T\ \}.
		$$
		Then, one may take
		$S_t(\om)=\om_t$ and $\F$ to be
		the filtration generated by $S$.}}$^{,}${\footnote{
		Note that we do not specify any probability measure.}}
with values in $\R^M_+$ for some $M$.
For every $\ell \in \cI$ 
there exist predictable integrands
$H_t\in \cB_b(\Omega,\cF_{t-1})$ for all 
$t=1,\ldots,T$ such that,
$$
\ell= \left(H \cdot S\right)_T:= \sum_{t=1}^{T} H_t\cdot \Delta S_t, \quad
{\text{where}}
\quad
\Delta S_t:=(S_{t}-S_{t-1}).
$$
Denote by $\ell_t:=(H\cdot S)_t$ for $t\in\cI$ and $\ell:=\ell_T$. 

Set $\hat{\ell}=\sum_{k,i} S_k^i -S_0^i$.
Then, one can directly show that  
with an appropriate $c^*$, we have  $L^*:=c^*+\hat{\ell} \ge 1$.  Define $\bdl$
using $\hat{\ell}$, set $\cH=\bdl$
and denote by $\cI_\ell$ 
the subset of $\cI$ with $H_t$ \emph{bounded} for every  $t=1,\ldots,T$.

We next prescribe the equivalence relation and the relevant sets.
Our starting point is the set of negligible sets $\cZ$
which we assume is given. 
We also make the following structural assumption.

\begin{asm}
	\label{a.discrete}  Assume that the trading
	is allowed only at finite time points labeled
	through $1,2,\ldots,T$.
	Let $\cI$ be 
	given as above and let
	$\cZ$ be a lattice which is closed 
	with respect to pointwise convergence.
	
	We also assume that
	$\cR=\cP^+$ 
	and the pre-order is given by,
	$$
	X\oleq Y\quad 
	\Leftrightarrow
	\quad \exists Z\in\cZ \ \ 
	{\text{ such that }}\ \ 
	X\pointleq Y+Z,
	$$
	where $\pointleq$ denotes the pointwise order of functions. 
	In particular, $X \in \cP$ if and only if there exists $Z\in \cZ$ such that $Z \pointleq X$.
	
\end{asm}

For an  example of the above structure, refer  to Example 2.6 of the main text.  In that example,
$\cZ$ consists of the  polar sets 
of a given  class $\cQ$ of probabilities. 
Then, in this context all inequalities
should be understood as $\cQ$ quasi-surely.
Also note also that the assumptions on $\cZ$ are trivially satisfied when $\cZ=\{0\}$. 
In this latter case, inequalities are pointwise.

Observe that in view
of the definition of $\oleq$ and the fact 
$\cR=\cP^+ $, 
$\ell \in \cI$ is an arbitrage if and only if 
there is $R^* \in \cP^+$ and $Z^* \in \cZ$,
so that
$\ell \pointgeq  R^*+ Z^*$.  
Hence, $\ell \in \cI$ is an arbitrage if and only
if $\ell \in \cP^+$.
We continue by showing the equivalence of the existence of an arbitrage
to the existence of a one-step arbitrage. 

\begin{lem}
	\label{lem: one step} Suppose that Assumption \ref{a.discrete} holds.  Then, there  
	exists arbitrage if and only if there exists 
	$t\in\{1,\ldots, T\}$, $h\in \cB_b(\Omega,\cF_{t-1})$ such that 
	$\ell:= h\cdot \Delta S_{t}$ 
	is an arbitrage.
\end{lem}
\begin{proof}
	The sufficiency is clear.
	To prove the necessity, suppose
	that $\ell \in \cI$ is an arbitrage.  Then, there is a predictable process $H$
	so that $\ell= (H \cdot S)_T$.
	Also $\ell \in \cP^+$, hence, $\ell\notin\cZ$ and there exists $Z\in\cZ$ such that 
	$\ell\ogeq Z$.  
	Define 
	$$
	\hat{t}:=\min\{t\in\{1,\ldots, T\}\ :\
	(H\cdot S)_t \in \cP^+\ \}\leq T.
	$$
	First we study the case where $\ell_{\hat{t}-1}\in\cZ$. Define 
	$$\ell^*:= H_{\hat{t}}\cdot \Delta S_{\hat{t}},
	$$
	and observe that $\ell_{\hat{t}}=\ell_{\hat{t}-1}+\ell^*$. Since $\ell_{\hat{t}-1}\in\cZ$, 
	we have that $\ell^*\in\cP^+$ iff $\ell_{\hat{t}}\in\cP^+$ and consequently 
	the lemma is proved.

	Suppose now $\ell_{\hat{t}-1}\notin\cZ$. If $\ell_{\hat{t}-1}\pointgeq 0$, 
	then $\ell_{\hat{t}-1}\in\cP$ and, thus, also in $\cP^+$, which is not possible 
	from the minimality of $\hat{t}$. Hence the set $A:=\{\ell_{\hat{t}-1}<_{\Omega} 0\}$
	is non empty and $\cF_{\hat{t}-1}$-measurable.
	Define, $h:=H_{\hat{t}}\chi_A$ and $\ell^*:= h\cdot \Delta S_{\hat{t}}$.
	Note  that, 
	$$
	\ell^*=\chi_A(\ell_{\hat{t}}-\ell_{\hat{t}-1})\pointgeq  \chi_A\ell_{\hat{t}}
	\pointgeq \chi_A Z\in\cZ.
	$$
	This implies $\ell^*\in\cP$. Towards a contradiction, suppose that $\ell^*\in \cZ$. Then,
	$$
	\ell_{\hat{t}-1}\pointgeq \chi_A\ell_{\hat{t}-1}\geq \chi_A\left(Z-\ell^*\right)\in\cZ,
	$$
	Since, by assumption, $\ell_{\hat{t}-1}\notin\cZ$ we have $\ell_{\hat{t}-1}\in\cP^+$
	from which $\hat{t}$ is not minimal.
\end{proof}

The following is the main result of this section. For the proof we follow 
the approach of  \cite{KabanovStricker} 
which is also used in \cite{BN13}. 
We consider the financial market
$\Theta_* =(\bdl,\|\cdot\|_\ell,\pointleq,\cI,\cP^+)$ described above.
\begin{thm}
	\label{t.tfae}
	In a finite discrete time
	financial market satifying the Assumption \ref{a.discrete}, the following are equivalent:
	\begin{enumerate}
		
		\item \label{AP1} The financial market $\Theta_*$ has no arbitrages.
		\item \label{AP2} The attainment property  holds and $\Theta_*$ is free
		of arbitrage.
		\item \label{AP3} The financial market $\Theta_*$ is strongly free of arbitrages.
	\end{enumerate}
\end{thm}

\begin{proof}
	In view of Proposition \ref{p.versus}
	we only need to prove the implication $\ref{AP1} \Rightarrow \ref{AP2}$. 
	
	For $X\in\cH$ such that $\cD(X)$ is finite we have that
	\begin{equation*}
		c_n+\cD(H)+\ell_n\pointgeq X+Z_n,
	\end{equation*} 
	for some $c_n\downarrow 0$, $\ell_n\in\cI$ and $Z_n\in\cZ$. Note that since $\cZ$ is a lattice we assume, without loss of generality, that $Z_n=(Z_n)^-$ and denote by $\cZ^-:=\{Z^-\mid Z\in\cZ\}$.\\
	
	We show that $\cC:=\cI-(\cL^0_+(\Om,\cF)+\cZ^-)$ is closed under pointwise 
	convergence where $\cL^0_+(\Om,\cF)$ denotes the class of pointwise nonnegative random variables. Once this result is shown, by observing that $X-c_n-\cD(X)=W_n\in\cC$ converges pointwise to $X-\cD(X)$ we obtain the attainment property.

	We  proceed by induction on the number of time steps. 
	Suppose first $T=1$. 
	Let
	\begin{equation}
		\label{e.convergence}
		W_n=\ell_n-K_n-Z_n\rightarrow W,
	\end{equation} where $\ell_n\in\cI$, $K_n\pointgeq 0$ and $Z_n\in\cZ^-$. 
	We need to show $W\in\cC$. Note that any $\ell_n$ can be 
	represented as $\ell_n=H_1^n\cdot\Delta S_1$ with $H_1^n\in\cL^0(\Omega,\cF_0)$.
	
	Let $\Omega_1:=\{\omega\in\Omega\mid \liminf |H_1^n|<\infty\}$. 
	From Lemma 2 in 
	\cite{KabanovStricker} 
	there exist a sequence 
	$\{\tilde{H}_1^k\}$ such that $\{\tilde{H}_1^k(\omega)\}$ 
	is a convergent subsequence of 
	$\{H_1^k(\omega)\}$ for every $\omega\in\Omega_1$. 
	Let $H_1:=\liminf H_1^n\ \chi_{\Omega^1}$ and $\ell:=H_1\cdot\Delta S_1$.

	Note now that $Z_n\pointleq 0$, hence, if $\liminf |Z_n|=\infty$ we have $\liminf Z_n=-\infty$. 
	We show that we can choose $\tilde{Z}_n\in\cZ^-$, $\tilde{K}_n\pointgeq 0$ such that $\tilde{W}_n:=\ell_n-\tilde{K}_n-\tilde{Z}_n\rightarrow W$ and $\liminf \tilde{Z}_n$ is finite on $\Omega_1$. On $\{\ell_n\pointgeq W\}$ set $\tilde{Z}_n=0$ and $\tilde{K}_n=\ell_n-W$. On $\{\ell_n <_{\Om} W\}$ set
	
	$$
	\tilde{Z}_n=Z_n\vee(\ell_n-W) ,\quad \tilde{K}_n=K_n\chi_{\{Z_n=\tilde{Z}_n\}}.
	$$
	It is clear that $Z_n\pointleq \tilde{Z}_n\pointleq 0$. From Lemma \ref{l.sandwich} we have $\tilde{Z}_n\in \cZ$. Moreover, it is easily checked that  $\tilde{W}_n:=\ell_n-\tilde{K}_n-\tilde{Z}_n\rightarrow W$. Nevertheless,
	from the convergence of $\ell_n$ on $\Omega_1$ and $\tilde{Z}_n\pointgeq -(W-\ell_n)^+$, 
	we obtain $\{\omega\in\Omega_1\mid \liminf \tilde{Z}_n> -\infty\}=\Omega_1$. As a consequence also $\liminf \tilde{K}_n$ is finite on $\Omega_1$, otherwise we could not have that $\tilde{W}_n\rightarrow W$. Thus, by setting $\tilde{Z}:=\liminf\tilde{Z}_n$ and $\tilde{K}:=\liminf\tilde{K}_n$, we have $W=\ell-\tilde{K}-\tilde{Z}\in\cC$.\\
	
	On $\Omega_1^C$ we may take $G_1^n:=H_1^n/|H_1^n|$ 
	and let $G_1:=\liminf G_1^n\chi_{\Omega_1^C}$. Define, $\ell_G:=G_1\cdot\Delta S_1$. We now observe that,
	$$
	\{\omega\in\Omega_1^C\mid \ell_G(\omega)\le 0\}
	\subseteq\{\omega\in\Omega_1^C\mid \liminf Z_n(\omega)=-\infty\}.
	$$ 
	Indeed, if $\omega\in\Omega_1^C$ is such that $\liminf Z_n(\omega)>-\infty$, applying again Lemma 2 in \cite{KabanovStricker},
	we have that 
	$$
	\liminf_{n \to \infty} \ \frac{X(\omega)+Z_n(\omega)}{
		|H_1^n(\omega)|}= 0,
	$$ 
	implying $\ell_G(\omega)$ is nonnegative. 
	Set now $$\tilde{Z}_n:=Z_n\vee-(\ell_G)^-.$$ 
	From $Z_n\pointleq \tilde{Z}_n\pointleq 0$, again by Lemma \ref{l.sandwich}, $\tilde{Z}_n\in \cZ$. 
	By taking the limit for $n\rightarrow\infty$ we obtain
	$(\ell_G)^-\in\cZ$ and thus, $\ell_G\in\cP$. 
	Since the financial market has no arbitrages $G_1\cdot\Delta S_1=Z\in\cZ$ 
	and hence one asset is redundant. Consider a partition 
	$\Omega_2^i$ of $\Omega_1^C$ on which $G_1^i\neq0$. Since $\cZ$ is stable under multiplication (Lemma \ref{l.stability}), for any $\ell^*\in\cI$, there exists $Z^*\in\cZ$ and $H^*\in\cL^0(\Omega_2^i,\cF_0)$ with $(H^*)^i=0$, such that $\ell^*=H^*\cdot \Delta S_1+Z^*$ on $\Omega_2^i$. Therefore, the term $\ell_n$ in \eqref{e.convergence} is 
	composed of trading strategies involving only $d-1$ assets. Iterating the procedure up to 
	$d$-steps we have the conclusion.

	Assuming now that \ref{e.convergence} holds for markets with $T-1$ periods, with the same argument we show that we can extend to 
	markets with $T$ periods. Set again  $\Omega_1:=\{\omega\in\Omega\mid \liminf |H_1^n|<\infty\}$. Since on $\Omega_1$ we have that, $$W_n-H_1^n\cdot\Delta S_1=\sum_{t=2}^T H_t^n\cdot \Delta S_t-K_n-Z_n\rightarrow W-H_1\cdot\Delta S_1.$$ The induction hypothesis allows to conclude that $W-H_1\cdot S_1\in\cC$ and therefore $W\in\cC$. On $\Omega_1^C$ we may take $G_1^n:=H_1^n/|H_1^n|$ 
	and let $G_1:=\liminf G_1^n\chi_{\Omega_1^C}$. Note that $W_n/|H_1^n|\rightarrow 0$ and hence $$
	\sum_{t=2}^T \frac{H_t^n}{|H_1^n|}\cdot \Delta S_t-\frac{K_n}{|H_1^n|}-\frac{Z_n}{|H_1^n|}\rightarrow -G_1\cdot\Delta S_1.$$
	Since $\cZ$ is stable under multiplication $\frac{Z_n}{|H_1^n|}\in\cZ$ and hence, by inductive hypothesis, there exists $\tilde{H}_t$ for $t=2,\ldots, T$ and $\tilde{Z}\in\cZ$ such that
	$$\tilde{\ell}:=G_1\cdot \Delta S_1+ \sum_{t=2}^T\tilde{H}_t\cdot \Delta S_t\pointgeq \tilde{Z}\in\cZ.$$
	The No Arbitrage condition implies that $\tilde{\ell}\in\cZ$. Once again, this means that one asset is redundant and, by considering a partition 
	$\Omega_2^i$ of $\Omega_1^C$ on which $G_1^i\neq0$, we can rewrite the term $\ell_n$ in \eqref{e.convergence} with $d-1$ assets. Iterating the procedure up to 
	$d$-steps we have the conclusion.
\end{proof}

The above result is consistent
with the fact that in classical
``probabilistic'' model for finite discrete-time markets only the no-arbitrage condition
and not the no-free lunch
condition has been 
utilized.

\section{Countably Additive Measures}
\label{s.ca}

In this section, we show that in general finite discrete time markets,
it is possible to characterize viability through countably additive functionals.
Also in this section, $\pointleq$ denotes the pointwise order for functions.
We prove this result by combining 
some results from \cite{Burzoni-et-al} which we collect in 
Section \ref{ss.ftm}. We refer to that paper for the precise technical requirements for $(\Omega,\F,S)$, 
we only point out that, in addition to the previous setting, $\Om$ needs to be a Polish space.

We let  $\cQ^{ca}$ be the set of countably additive
positive  probability measures $\Q$, with finite support, such that
$S$ is a $\Q$-martingale and $\cZ^-:=\{-Z^-\mid Z\in \cZ\}$.  For $X \in \cH$, set
$$
\cZ(X):= \left\{ Z \in \cZ^-\ :\ \exists \ell \in \cI \quad
{\text{such that}}\quad \cD(X)+\ell \pointgeq X+Z\ \right\},
$$
which is always non-empty when $\cD(X)$, e.g. $\forall X\in\bdd$.
By the lattice property of $\cZ$, if $\cD(X)+\ell \pointgeq X+Z$ the same is true if we take 
$Z=Z^-$. From Theorem \ref{t.tfae} we know
that, under no arbitrage, the attainment property holds and, hence, 
$\cZ(X)$ is non-empty for every $X \in \cH$.
For $A \in \cF$, we define
\begin{align*}
	\cD_A(X)&:= \inf \left\{ c \in \R\ :
	\exists \ell \in \cI  \quad
	{\text{such that}}\quad c+\ell(\om) \ge X(\om), \ \forall \om \in A\ 
	\right\}\\
	\cQ^{ca}_A&:= \left\{ \Q \in \cQ^{ca}\ :\
	\Q(A)=1\ \right\}.
\end{align*}

We need the following technical result in the 
proof of the main Theorem. 

\begin{prp}
	\label{p.technical}
	Suppose Assumption \ref{a.discrete} holds
	and  the financial market has no arbitrages.  Then, for every
	$X \in \cH$ and $Z \in \cZ(X)$, there exists
	$A_{X,Z}$ such that
	\begin{equation}\label{eq:technical set}
		A_{X,Z} \subset \{ \ \om \in \Om \ :\
		Z(\om)=0 \ \},
	\end{equation}
	and
	$$
	\cD(X)= \cD_{A_{X,Z}}(X) = \sup _{\Q \in \cQ^{ca}_{A_{X,Z}}}\ \E_\Q[X].
	$$
\end{prp}

Before proving this result,
we state the main result of this section.

\begin{thm}
	\label{t.FTAP}
	Suppose Assumption \ref{a.discrete} holds.
	Then,  the financial market has no arbitrages
	if and only if for every $(Z,R) \in \cZ^- \times \cP^+$
	there exists $\Q_{Z,R} \in \cQ^{ca}$ satisfying
	\begin{equation}\label{eq: martingale ZR}
		\E_{\Q_{Z,R}}[R]>0
		\quad
		{\text{and}}
		\quad
		\E_{\Q_{Z,R}}[Z]=0.
	\end{equation}
\end{thm}
\begin{proof}
	Suppose that the financial market has no arbitrages.  Fix
	$(Z,R) \in \cZ^- \times \cP^+$ and $Z_R \in \cZ(R)$. Set  $Z^*:=Z_R+Z \in \cZ(R)$.
	By Proposition \ref{p.technical},
	there exists $A_*:=A_{R,Z^*}$ satisfying
	the properties listed there.  In particular,
	$$
	0< \cD(R) = \sup_{\Q \in \cQ^{ca}_{A_*}}\ \E_\Q[R].
	$$
	Hence, there is $\Q^* \in  \cQ^{ca}_{A_*}$ so that
	$\E_{\Q^*}[R]>0$.  Moreover, since $Z_R,Z \in \cZ^-$,
	$$
	A_* \subset \{ Z^*=0\} = \{Z_R=0 \} \cap \{Z=0\}.
	$$
	In particular, $\E_{\Q^*}[Z]=0$.
	
	To prove the opposite implication, suppose that there exists $R\in\cP^+$,  
	$\ell \in \cI$ and $Z\in \cZ$ such that $\ell \pointgeq R+Z$.
	Then, it is clear that $\ell \pointgeq R-Z^-$. Let $\Q^*:=\Q_{-Z^-,R}\in\cQ^{ca}$ 
	satisfying \eqref{eq: martingale ZR}. By integrating
	both sides against $\Q^*$, we obtain
	$$
	0 = \E_{\Q^*}[\ell] \ge\E_{\Q^*}[R-Z^-] =
	\E_{\Q^*}[R]>0.
	$$
	which is a contradiction. Thus, there are no arbitrages.
\end{proof}
\vspace{5pt}

We continue with 
the proof of Proposition \ref{p.technical}.

\begin{proof}[Proof of Proposition \ref{p.technical}]
	Since there are no arbitrages,
	by Theorem \ref{t.tfae} we have the attainment property.
	Hence, for a given $X\in\cH$, the set $\cZ(X)$ is non-empty.\\
	
	\emph{Step 1.} We show that, for any $Z\in\cZ(X)$, $\cD(X)=\cD_{\{Z=0\}}(X)$.\\
	Note that, since $\cD(X)+\ell\pointgeq X+Z$, for some $\ell\in\cI$, the inequality 
	$\cD_{\{Z=0\}}(X)\le\cD(X)$ is always true. 
	Towards a contradiction, suppose that the inequality is strict, 
	namely, there exist $c<\cD(X)$ and $\tilde{\ell}\in \cI$ such that 
	$c+\tilde{\ell}(\omega)\geq X(\omega)$ for any $\omega\in\{Z=0\}$. 
	We show that 
	$$
	\tilde{Z}:=(c+\tilde{\ell}-X)^-\chi_{\{Z<0\}}\in \cZ.
	$$ 
	This together with $c+\tilde{\ell}\pointgeq X+\tilde{Z}$ yields a contradiction. 
	Recall that $\cZ$ is a linear space so that $nZ\in\cZ$ for any $n\in\mathbb{N}$. 
	From $nZ\pointleq \tilde{Z}\vee (nZ)\pointleq 0$, we also have $\tilde{Z}_n:=\tilde{Z}\vee (nZ)\in \cZ$, 
	by Lemma \ref{l.sandwich}. By noting that $\{\tilde{Z}<0\}\subset\{Z<0\}$ 
	we have that $\tilde{Z}_n(\omega)\rightarrow \tilde{Z}(\omega)$ for every $\omega\in\Omega$. 
	From the closure of $\cZ$ under pointwise convergence, we conclude that $\tilde{Z}\in \cZ$.
	
	\emph{Step 2.} For a given set $A\in\cF_T$, we let $A^*\subset A$ 
	be the set of scenarios visited by martingale measures 
	(see \eqref{eq:OmegaStar_def} in the Appendix for more details). 
	We show that, for any $Z\in\cZ(X)$, $\cD(X)=\cD_{\{Z=0\}^*}(X)$.\\
	
	Suppose that $\{Z=0\}^*$ is a proper subset of $\{Z=0\}$ otherwise, from Step 1, 
	there is nothing to show. From Lemma \ref{NoOnePoint} there is a strategy 
	$\tilde{\ell}\in \cI$ such that $\tilde{\ell}\geq 0$ on 
	$\{Z=0\}$\footnote{Note that restricted to $\{Z=0\}$ this strategy yields no risk 
		and possibly positive gains, in other words, this is a good candidate for being an arbitrage.}. 
	Lemma \ref{conditional:splitting} (and in particular \eqref{e.A0}) 
	yields a finite number of strategies $\ell^t_1,\ldots \ell^t_{\beta_t}$ with $t=1,\ldots T$, such that 
	\begin{equation}\label{eq:ZOmega star}
		\{\hat{Z}=0\}=\{Z=0\}^*\qquad \text{where }\qquad \hat{Z}:=Z-\sum_{t=1}^T\sum_{i=1}^{\beta_t}\chi_{\{Z=0\}}(\ell^t_i)^+\ .
	\end{equation}
	Moreover, for any $\omega\in\{Z=0\}\setminus \{Z=0\}^*$, 
	there exists $(i,t)$ such that $\ell^t_i(\omega)>0$.
	We are going to show that, under the no arbitrage hypothesis, 
	$\ell^t_i\in\cZ$ for any $i=1,\ldots\beta_t$, $t=1,\ldots T$. 
	In particular, from the lattice property of the linear space $\cZ$, we have $\hat{Z}\in\cZ$.
	
	We illustrate the reason for $t=T$, by repeating the same argument up to $t=1$ we have the thesis. 
	We proceed by induction on $i$. Start with $i=1$. From Lemma \ref{conditional:splitting} we have that 
	$\ell^T_i\geq 0$ on $\{Z=0\}$ and, therefore, $\{\ell^T_1<0\}\subseteq\{Z<0\}$. 
	Define $\tilde{Z}:=-(\ell^T_1)^-\pointleq 0$. By using the same argument as in Step 1, 
	we observe that $nZ\pointleq \tilde{Z}\vee (nZ)\pointleq 0$ with $nZ\in\cZ$ for any $n\in\mathbb{N}$. 
	From $\{\ell^T_1<0\}\subseteq\{Z<0\}$ and the closure of $\cZ$ under pointwise convergence, 
	we conclude that $\tilde{Z}\in \cZ$. From no arbitrage, we must have $\ell^T_1\in\cZ$.\\
	Suppose now that $\ell_j^T\in\cZ$ for every $1\leq j\leq i-1$.
	From Lemma \ref{conditional:splitting}, we have that $\ell^T_i\geq 0$ on $\{Z-\sum_{j=1}^{i-1}\ell^T_i=0\}$ and, 
	therefore, $\{\ell^T_i<0\}\subseteq\{Z-\sum_{j=1}^{i-1}\ell^T_i<0\}$.  
	The argument of Step 1 allows to conclude that $\ell^T_i\in\cZ$.\\ 
	
	We are now able to show the claim. The inequality $\cD_{\{Z=0\}^*}(X)\le \cD_{\{Z=0\}}(X)=\cD(X)$ is always true. 
	Towards a contradiction, suppose that the inequality is strict, namely, there exist 
	$c<\cD(X)$ and $\tilde{\ell}\in \cI$ such that $c+\tilde{\ell}(\omega)\geq X(\omega)$ 
	for any $\omega\in\{Z=0\}^*$. We show that 
	$$
	\tilde{Z}:=(c+\tilde{\ell}-X)^-\chi_{\Omega\setminus \{Z=0\}^*}\in \cZ.
	$$ 
	This together with $c+\tilde{\ell}\pointgeq X+\tilde{Z}$, yields a contradiction.
	To see this recall that, from the above argument, 
	$\hat{Z}\in\cZ$ with $\hat{Z}$ as in \eqref{eq:ZOmega star}. 
	Moreover, again by \eqref{eq:ZOmega star}, we have $\{\tilde{Z}<0\}\subset\{\hat{Z}<0\}$. 
	The argument of Step 1 allows to conclude that $\tilde{Z}\in \cZ$.
	
	\emph{Step 3.} We are now able to conclude the proof. Fix $Z\in\cZ(X)$ and set $A_{X,Z}:=\{Z=0\}^*$. Then,
	$$
	\cD(X) =\cD_{\{Z=0\}}(X)
	=\cD_{(A_{X,Z})^*}(X)
	= \sup_{Q\in \mathcal{Q}^{ca}_{A_{X,Z}}}\E_{Q}[X],
	$$
	where the first two equalities follow from Step 1 and Step 2 and the 
	last equality follows from Proposition \ref{thm: duality_no_option}. 
\end{proof}

\section{Some technical tools}
\subsection{Preferences}
\label{ss.pre}

We start with a simple but  a
useful condition for negligibility.

\begin{lem}
	\label{l.sandwich}  Consider two negligible claims
	$\hat{Z},\tilde{Z}\in\cZ$.
	Then, any claim $Z\in\cH$ 
	satisfying $\hat{Z}\oleq Z\oleq \tilde{Z}$  is 
	negligible as well.
\end{lem}
\begin{proof}
	By definitions, we have, 
	$$
	X\oleq X+\hat{Z}\oleq X+Z \oleq X+ \tilde{Z}\oleq X \ \
	\Rightarrow \ \
	X \sim X+Z.
	$$
	Thus, $Z\in \cZ$. 
\end{proof}

\begin{lem}\label{l.stability} Suppose that $\cZ$ is closed under 
	pointwise convergence. 
	Then, $\cZ$ is stable under multiplication, i.e., $ZH\in\cZ$ for any $H\in\cH$.
\end{lem}
\begin{proof}Note first that $Z_n:=Z((H\wedge n)\vee -n)\in\cZ$. This follows from by 
	Lemma \ref{l.sandwich} and the fact that $\cZ$ is a cone. 
	By taking the limit  for $n\rightarrow\infty$, the result follows.
\end{proof}

We next prove that $\cE(Z)=0$ for every
$Z \in \cZ$.

\begin{lem}
	\label{l.ez} Let $\cE$ be a
	sublinear expectation. Then,
	\begin{align}
		\label{e.e}
		\cE(c + \lambda [X+Y]) &= c + \cE(\lambda[X+Y])
		= c + \lambda \cE(X+Y)\\
		\nonumber& \le c + \lambda \left[\ - \left(
		- \cE(X) - \cE(Y)\right) \right],
	\end{align}
	for every $c \in $, $\lambda \ge 0$, $X,Y \in 
	\cH$.
	In particular,  
	\begin{equation*}
		\cE(Z)=0, \quad \forall \ Z \in \cZ.
	\end{equation*}
\end{lem}
\begin{proof}
	Let $X, Y\in \cH$.  The sub-additivity 
	of $U_\cE$ implies that
	$$
	U_\cE(X^\prime)+U_\cE(Y^\prime) \le U_\cE(X^\prime
	+Y^\prime), \quad \forall \ X^\prime, Y^\prime \in \cH,
	$$
	even when they take values $\pm \infty$.  The definition
	of $U_\cE$ now yields,
	$$
	\cE(X+Y) = - U_\cE(-X-Y) \le 
	-\left[U_\cE(-X)+U_\cE(-Y)\right] = - \left(
	- \cE(X) - \cE(Y)\right).
	$$
	Then, \reff{e.e} follows directly
	from the definitions. 
	
	Let $Z \in \cZ$.  Then, $-Z , Z\in \cP$ and 
	$\cE(Z), \cE(-Z) \ge 0$.  
	Since $-Z \in \cP$, the
	monotonicity of $\cE$
	implies that 
	$\cE(X-Z) \ge \cE(X)$ for
	any $X \in \cH$.  Choose $X=Z$
	to arrive at
	$$
	0 = \cE(0)=\cE(Z-Z)\ge \cE(Z) \ge 0.
	$$
	Hence,  $\cE(Z)$ is equal to zero.
\end{proof}

\subsection{Finite Time Markets}
\label{ss.ftm}

We here recall some results from \cite{Burzoni-et-al} 
(see Section 2 therein for the precise specification of the framework). 
We are given a filtered space $(\Omega,\F,\mathcal{F})$ with 
$\Omega$ a Polish space and $\F$ containing the   filtration  generated by  a Borel-measurable process $S$. 
We denote by $\mathcal{Q}$ the set of martingale measures for the process $S$, 
whose support is a finite number of points. For a given set $A\in\mathcal{F}$, 
$\mathcal{Q}_A=\{Q\in\mathcal{Q}\mid Q(A)=1\}$. 
We define the set of scenarios charged by martingale measures as
\begin{equation}
	\label{eq:OmegaStar_def}
	A^*  := \left\{ \omega \in \Omega \mid \exists Q\in \mathcal{Q}_A\text{ s.t. }Q(\omega )>0\right\} 
	=\bigcup_{Q\in \mathcal{Q}_A} supp(Q).
\end{equation}

\begin{dfn}
	We say that $\ell\in\cI$ is a one-step strategy if $\ell=H_t\cdot (S_{t}-S_{t-1})$ with 
	$H_t\in \mathcal{L}(X,\mathcal{F}_{t-1})$ for some $t\in \{1,\ldots,T\}$. We say that 
	$a\in \cI$ is a one-point Arbitrage on $A$ iff $a(\omega)\geq 0$ $\forall\omega\in A$ 
	and $a(\omega)>0$ for some $\omega\in A$.
\end{dfn}
The following Lemma is crucial for the characterization of the set $A^*$ in terms of arbitrage considerations.

\begin{lem}
	\label{conditional:splitting} Fix any $t\in \{1,\ldots ,T\}$ and $\Gamma \in
	\mathcal{F}$. There exist an index $\beta \in \{0,\ldots ,d\}$%
	, one-step strategies $\ell^{1},\ldots ,\ell^{\beta }\in \cI$ and  $B^{0},...,B^{\beta }$, a partition of $\Gamma$, satisfying:
	
	\begin{enumerate}
		\item \label{item_spezz} if $\beta=0$ then $B^0=\Gamma$ and there are No one-point Arbitrages, i.e., 
		$$
		\ell(\omega)\geq 0\  \forall \omega\in B^{0}\Rightarrow\ell(\omega)=0\ \forall \omega\in B^{0}.
		$$
		\item \label{item_spezz_arb}if $\beta >0$ and $i=1,\ldots ,\beta $ then: 
		\begin{itemize}
			\item [$\rhd$]$B^{i}\neq \emptyset $,
			\item [$\rhd$]$\ell^{i}(\omega )>0$ for all $\omega \in B^{i}$,
			\item [$\rhd$]$\ell^{i}(\omega )\geq 0$ for all $\omega \in \cup _{j=i}^{\beta }B^{j}\cup B^{0 }$.
		\end{itemize}
	\end{enumerate}
\end{lem}

We are now using the previous result, which is for some fixed $t$, to identify $A^*$. Define
\begin{eqnarray}\label{eq:construction A*}
	A_{T}:= && A \nonumber \\
	A _{t-1}:= && A _{t}\setminus \bigcup_{i=1}^{\beta_t} B^i_t,\text{%
		\quad }t\in \{1,\ldots,T\}, \label{Omega_t}
\end{eqnarray}%
where $B^i_t:= B^{i,\Gamma}_t$, $\beta_t:=\beta_t^{\Gamma}$ are the sets and
index constructed in Lemma \ref{conditional:splitting} with $\Gamma= A_t$%
, for $1\leq t\leq T$. 
Note that, for the corresponding strategies $\ell^t_i$ we have 
\begin{equation}\label{e.A0}
	A_0=
	\bigcap_{t=1}^T\bigcap_{i=1}^{\beta_t}\{\ell^t_i=0\}.
\end{equation}

\begin{lem}\label{NoOnePoint} $A_0$ as constructed in \eqref{eq:construction A*} satisfies 
	$ A_0= A^*$. Moreover, No one-point Arbitrage on $ A$ $\Leftrightarrow$ $ A^*= A$.
\end{lem}

\begin{prop}
	\label{thm: duality_no_option} Let $ A \in
	\mathcal{F}$. We have that for any $\mathcal{F}$-measurable random variable $g$,
	\begin{equation}
		\pi _{ A ^{\ast }}(g)=\sup_{Q\in \mathcal{Q}_A}\E_{Q}[g]. \label{super:0opt}
	\end{equation}%
	with $\pi _{ A ^{\ast }}(g)=\inf \left\{ x\in \mathbb{R}\mid
	\exists
	a\in \cI \text{ such that }%
	x+a_{T}(\omega)\geq g(\omega)\ \forall \omega\in A ^{*}\right\} $. 
	In particular, the left hand side of
	\eqref{super:0opt} is attained by some strategy $a\in\cI$.
\end{prop}
\end{appendix}

\bibliographystyle{econometrica}
\bibliography{viable-5}

\end{document}